\pdfoutput=1
\RequirePackage{ifpdf}
\ifpdf % We~are running pdfTeX in pdf mode
\documentclass[pdftex]{sigma}%,draft
\else
\documentclass{sigma}
\fi

\numberwithin{equation}{section}

\def\p{\partial}
\def\G{\Gamma}
\def\g{\gamma}
\def\s{\sigma}

\def\zb{\bar z}

\def\a{\alpha}

\def\A{{\mathcal A}}
\def\B{{\mathcal B}}
\def\C{{\mathcal C}}

\def\D{{\mathcal D}}

\def\cL{{\mathcal L}}
\def\K{{\cal K}}

\def\K{{\mathcal K}}

\def\P{{\mathcal P}}
\def\X{{\mathcal X}}

\def\res{\operatorname{res}}

\newtheorem{theo}{Theorem}[section]

\newtheorem{cor}[theo]{Corrolary}
\newtheorem{lem}[theo]{Lemma}

\theoremstyle{definition}
\newtheorem{dfn}[theo]{Definition}
\newtheorem{rem}[theo]{Remark}

\newcommand {\BC} {\mathbb C}

\newcommand {\BR} {\mathbb R}
\newcommand {\BP} {\mathbb P}

\newcommand {\BZ} {\mathbb Z}

\newcommand {\CalA} {\mathcal A}

\newcommand {\CalK} {\mathcal K}

\newcommand {\CalM} {\mathcal M}
\newcommand {\CalN} {\mathcal N}
\newcommand {\CalO} {\mathcal O}

\newcommand{\fo}{\vert\kern -.03in\_}%{\guilsinglleft}

\newcommand {\ii} {\mathrm{i}}

\newcommand{\pa}{{\partial}}
\newcommand{\pb} {{\bar\partial}}

\begin{document}
%\allowdisplaybreaks

\renewcommand{\thefootnote}{}

\newcommand{\arXivNumber}{2106.14201}

\renewcommand{\PaperNumber}{006}

\FirstPageHeading

\ShortArticleName{Novikov--Veselov Symmetries of the Two-Dimensional $O(N)$ Sigma Model}

\ArticleName{Novikov--Veselov Symmetries\\ of the Two-Dimensional $\boldsymbol{O(N)}$ Sigma Model\footnote{This paper is a~contribution to the Special Issue on Mathematics of Integrable Systems: Classical and Quantum in honor of Leon Takhtajan.

~~\,The full collection is available at \href{https://www.emis.de/journals/SIGMA/Takhtajan.html}{https://www.emis.de/journals/SIGMA/Takhtajan.html}}}

\Author{Igor KRICHEVER~$^{\rm acd}$ and Nikita NEKRASOV~$^{\rm bce}$}

\AuthorNameForHeading{I.~Krichever and N.~Nekrasov}

\Address{$^{\rm a)}$~Department of Mathematics, Columbia University, New York, USA}
\EmailD{\href{mailto:krichev@math.columbia.edu}{krichev@math.columbia.edu}}

\Address{$^{\rm b)}$~Simons Center for Geometry and Physics, Stony Brook University, Stony Brook NY, USA}
\EmailD{\href{mailto:nnekrasov@scgp.stonybrook.edu}{nnekrasov@scgp.stonybrook.edu}}

\Address{$^{\rm c)}$~Center for Advanced Studies, Skoltech, Russia}

\Address{$^{\rm d)}$~Higher School of Economics, Moscow, Russia}

\Address{$^{\rm e)}$~Kharkevich Institute for Information Transmission Problems, Moscow, Russia}

\ArticleDates{Received October 19, 2021; Published online January 24, 2022}

\Abstract{We show that Novikov--Veselov hierarchy provides a complete family of com\-muting symmetries of two-dimensional $O(N)$ sigma model. In the first part of the paper we use these symmetries to prove that the Fermi spectral curve for the double-periodic sigma model is algebraic. Thus, our previous construction of the complexified harmonic maps in the case of irreducible Fermi curves is complete. In the second part of the paper we generalize our construction to the case of reducible Fermi curves and show that it gives the \emph{conformal} harmonic maps to \emph{even}-dimensional spheres. Remarkably, the solutions are parameterized by spectral curves of \emph{turning points} of the elliptic Calogero--Moser system.}


\Classification{14H70; 17B80; 35J10; 37K10; 37K20; 37K30; 81R12}

\begin{flushright}
\begin{minipage}{53mm}
\it To Leon Takhtajan on occasion\\ of his 70$\,{}^{th}$ birthday
\end{minipage}
\end{flushright}

\renewcommand{\thefootnote}{\arabic{footnote}}
\setcounter{footnote}{0}

\section{Introduction}

Harmonic maps of two-dimensional Riemann surface $\Sigma$ to a Riemann manifold $M$
are of interest both in physics and mathematics. The harmonic maps are the critical points of the Dirichlet functional, the sigma model action
\begin{gather} \label{action}
S(X)=\int_{\Sigma} \sqrt{h} h^{ab} g_{ij}(X)\p_{a} X^i\p_{b}X^j \, {\rm d}x {\rm d}y = \int_{\Sigma} g_{ij}(X){\pa} X^i {\bar\pa} X^j .
\end{gather} Here the map ${\Phi} \colon \Sigma\to M$ is represented, locally, by the pullbacks $X^i(x, y) = {\Phi}^{*} x^i$, $i = 1, \dots, \dim M$, of the coordinate functions $(x^{i})$ on $M$ to $\Sigma$, with $x$, $y$ local real coordinates on~$\Sigma$. In~\eqref{action} $(x, y)$ denote the real coordinates on $\Sigma$, while ${\pa} = {\pa}_{z}$, ${\pb} = {\p}_{\zb}$ denote the $z$ and $\zb$-derivatives in the complex structure
determined by the conformal class of the metric $h_{ab}$ via $h_{ab} {\rm d}y^{a}{\rm d}y^{b} \propto {\rm d}z{\rm d}{\bar z}$. Finally, $g_{ij}(X) {\rm d}X^i {\rm d}X^j$ is a Riemann metric on the target manifold~$M$.

\begin{rem}
Assume $M$ is a real slice of a complex manifold $M_{\BC}$ whose metric $g_{ij}(X)\,{\rm d}X^i {\rm d}X^j$ admits an analytic continuation to $M_{\BC}$ as a holomorphic symmetric $(2,0)$-tensor field $g_{\BC}$ with positive definite real part. Likewise, we can allow for the complex metric $h_{\BC} = h_{ab}\,{\rm d}y^a {\rm d}y^b$ on the worldsheet $\Sigma$, which defines two conformal structures: one by $h_{\BC} ({\pa}_z, {\pa}_z) = 0$, another by $h_{\BC} ({\pa}_{\zb}, {\pa}_{\zb}) = 0$.
For real $h$ these structures are complex conjugate to each other (as points in~${\CalM}_{g}$). In what follows we shall make these abstract constructions quite explicit in the case of a two dimensional torus $\Sigma = S^{1} \times S^{1}$.

In case $M$ admits a group $G$ of isometries, extending, possibly infinitesimally, to the action of $\operatorname{Lie}(G_{\BC})$ on $M_{\BC}$ preserving $g_{\BC}$, the action~\eqref{action} can be deformed to
\begin{gather}
\label{eq:twaction}
S_{h_{\BC}, A_{\BC}}(X)=\int_{\Sigma} \sqrt{h_{\BC}} h_{\BC}^{ab} (g_{\BC})_{ij}(X){\nabla}_{a} X^i {\nabla}_{b}X^j \, {\rm d}x{\rm d}y,
\end{gather}
where ${\nabla}_a X^{i} = {\pa}_{a} X^{i} + (A_{\BC})_{a}^{\rm A} V_{\rm A}^{i}(X)$ is the covariant derivative in the background gauge field~$A_{\BC}$ on~$\Sigma$, a connection on a principal $G_{\BC}$-bundle $\mathcal P$ over $\Sigma$. The action~\eqref{eq:twaction} then describes the harmonic sections of the bundle
\[
M_{\BC} \times_{G_{\BC}} {\mathcal P}
\]
over $\Sigma$. We call this problem a \emph{twisted sigma model}. We only consider the case of flat $G_{\BC}$-connections,
\[
{\pa}_{a} (A_{\BC})_{b} - {\pa}_{b} (A_{\BC})_{a} + [ (A_{\BC})_{a}, (A_{\BC})_{b} ] = 0
\]
 on a topologically trivial bundle ${\mathcal P} = {\Sigma} \times G_{\BC}$.
\end{rem}

\begin{rem} This paper is a continuation of~\cite{KN}, where the most important case is the complexified twisted one. In particular, it motivates the Fermi-curve approach we also employ presently. Nevertheless, in this paper, mostly, we study the real untwisted case.
\end{rem}

The zero-curvature representation of the equation of motion for the two-dimensional $O(N)$ sigma-model, whose target manifold is the sphere $M=S^{N-1}$ with the metric induced by a~standard imbedding into~${\BR}^N$, was discovered by Pohlmeyer in~\cite{pohl} where an infinite series of conservation laws was constructed.
In~\cite{mikh-zakh} these results were extended to a wider class of models including the principal chiral field model~-- the sigma model with the target space $M = G$, a Lie group with the bi-invariant metric $\operatorname{tr} \big(X^{-1}{\rm d}X\big)^{2}$.

In \cite{kr-dalamber} it was shown that {\it locally} the principal chiral field equations are integrable in the following strong sense: their solutions admit a non-linear analog of the d'Alembert representation as a superposition of functions depending only on one variable. The problem of constructing the {\it global}
solutions, i.e., solutions defined on a compact Riemann surface is much harder.

\begin{rem} The models studied in \cite{kr-dalamber, pohl,mikh-zakh} were the sigma models with Minkowski source. Accordingly, the equations of motion of these models are hyperbolic. The equations of motion of the sigma models studied in the present paper are elliptic. The equations of motion and their zero-curvature representation
are easily transformed to the elliptic case by the simple change of the light-cone variables
$\xi_{\pm} = x \pm t$ to the holomorphic-antiholomorphic variables~$z$,~$\zb$.
However, the nature of solutions in the two cases is quite different, the compactness of the source $\Sigma$ playing a key r{\^o}le.
\end{rem}

\subsection{Worldsheet is a two-torus}

Our main interest is in the double-periodic two-dimensional sigma model, i.e., we take $\Sigma$ to be a~two-dimensional torus $T^2 = {S}^{1} \times {S}^{1}$. Its conformal structure is parameterized by the complex number $\tau$, with $\operatorname{Im} \tau >0$, which usually is introduced by representing ${\Sigma} \approx {\BC}/{\BZ} \oplus {\tau}{\BZ}$,
\begin{gather}
\begin{split}
&\Sigma= {\BR}^{2} / {\BZ}^{2} = \{(z,\zb)\sim (z+\omega_\alpha,\zb+\bar \omega_\alpha) \, | \, {\a} = x, y \} , \\
& \omega_x= {\bar\omega}_{x} = 1 , \qquad \omega_y=\tau , \qquad {\bar\omega}_{y} = {\bar\tau}.
\end{split}\label{sigma}
\end{gather}
In the complexified approach we have the complex metric
\begin{gather}
h_{\BC} = \frac{( {\rm d}x + {\tau} {\rm d}y )( {\rm d}x + {\bar\tau} {\rm d}y )}{{\tau}_2},
\label{eq:eucl}
\end{gather}
whose conformal class is parametrized by two complex numbers ${\tau}_{1}$, ${\tau}_{2}$, with $\operatorname{Re}({\tau}_{2}) > 0$, or, equivalently, by
\[
{\tau} = {\tau}_{1} + {\ii}{\tau}_{2} , \qquad {\bar\tau} = {\tau}_{1} - {\ii}{\tau}_{2}.
\]
We shall call the real two-torus $T^2$ endowed with the complex metric \eqref{eq:eucl} with $\operatorname{Re}({\tau}_{2}) > 0$ a~{\it Euclidean torus $T^2$} in what follows.

The holomorphic vector field ${\pa}_{z}$ and the antiholomorphic vector field ${\pa}_{\zb}$ are given by
\begin{gather*}
{\pa}_{z} = \frac{{\pa}_{y} - {\bar\tau} {\pa}_{x}}{2{\ii}{\tau}_2} , \qquad
{\p}_{\zb} = \frac{{\pa}_{y} - {\tau} {\pa}_{x}}{-2{\ii}{\tau}_2}
%\label{eq:holantiholvf}
\end{gather*}
and the (${\BC}$-linear) Hodge star acts on $1$-forms on $\Sigma$ via
\begin{alignat*}{3}
&\star ( {\rm d}x + {\tau}{\rm d} y) = {\ii} ({\rm d}x + {\tau} {\rm d}y) , \qquad&& \star ( {\rm d}x + {\bar\tau} {\rm d}y ) = - {\ii} ({\rm d}x + {\bar\tau} {\rm d}y ) , &\\
 &\star {\rm d}x = - \frac{ {\tau}_{1} {\rm d}x + {\tau} {\bar\tau} {\rm d}y}{{\tau}_2} , \qquad &&\star {\rm d}y = \frac{{\rm d}x + {\tau}_{1} {\rm d}y}{{\tau}_2}.&
%\label{eq:hodgec}
\end{alignat*}
In what follows we shall use the notations $T_{x}$, $T_{y}$ for the translations
\begin{gather}
T_{x} \colon \ (x,y) \mapsto (x+1, y) , \qquad T_{y}\colon \ (x,y) \mapsto (x, y+1)
\label{eq:deck}
\end{gather}
of the universal cover ${\tilde\Sigma} = {\BR}^{2}$ of ${\Sigma} = {\tilde\Sigma}/{\BZ}^{2}$, also known as the deck transformations.

Below, in using the derivatives along $\Sigma$, we either use specifically $(1,0)$ and $(0,1)$-derivati\-ves~${\pa}_{z}$ and~${\pa}_{\zb}$, or use the exterior derivative
\[
{\rm d}_{\Sigma} = {\rm d}x {\pa}_{x} + {\rm d}y {\pa}_{y} = {\rm d}z {\pa}_{z} + {\rm d}{\zb} {\pa}_{\zb} .
\]
The differentials ${\rm d}w_{\a}$, ${\rm d}k$, ${\rm d}{\Omega}$ etc.\ denote the $(1,0)$-differentials along the Fermi-curve.

\subsection{Spectral curves from zero-curvature equations}
The harmonic maps of the two-torus $T^2$ to $S^3$ were constructed in \cite{hitchin90} via the zero-curvature representation for the principal chiral ${\rm SU}(2)$ model. The latter is the compatibility condition for the system of two linear equations
\begin{gather}\label{l-sys}
\left(\p_z-\frac {U(z,\zb)}{\lambda+1}\right)\Psi(z,\zb,\lambda)=0 ,\qquad
\left(\p_{\zb}+\frac{V(z,\zb)}{\lambda-1}\right)\Psi(z,\zb,\lambda)=0
\end{gather}
with $U=X^{-1}\p_z X$ and $V=X^{-1}\p_{\zb} X$. Since the fundamental group of a torus is abelian the operators
\begin{gather*}%\label{monod}
B_\alpha(\lambda):= ( T_{\alpha}^{*} \Psi(\lambda) ) \Psi^{-1}(\lambda) , \qquad \alpha= x, y
\end{gather*}
of monodromy of flat connection \eqref{l-sys} commute, $[ B_{x}({\lambda}) , B_{y}({\lambda}) ] = 0$. Therefore, the branch points of the two-sheeted covers (for the ${\rm SU}(2)$ case the matrices $B_\a$ are of rank 2) of the complex $\lambda$-plane defined by the characteristic equations
\begin{gather}\label{charB}
R_{\a} (\mu_{\a},\lambda):=\det \left(\mu_{\a} \cdot \mathbb I-B_\alpha(\lambda)\right)=0
\end{gather}
coincide. The key step in~\cite{hitchin90} was the proof that the number of these branch points is
\emph{finite}. The hyperelliptic curve defined by these branch points is a normalization of the analytic spectral curves~\eqref{charB}.

In general, with this construction, given the hyperelliptic curve and a point on its Jacobian, we get a quasi-periodic harmonic map of the universal cover of $\Sigma$. Solving the \emph{periodicity} constraint is a pain. Indeed, the double-periodic maps are singled out by certain conditions on periods of certain abelian differentials of the second kind on the hyperelliptic curve.
These equations are transcendental and so hard to control that the genus of the hyperelliptic curves in the examples found in~\cite{hitchin90} is {\it at most} three.

For further comparison with the results of that work, we recall one more observation made in~\cite{hitchin90}: \emph{there is only one class of harmonic maps which can not be reconstructed by the algebraic-geometrical data, and those are the ones for which the Floquet multipliers $\mu_\a$ are constants}.
These maps are the conformal maps to a \emph{totally geodesic} $2$-sphere in $S^3$ which up to isometry of $S^3$ is the locus
\[ \big\{ g^2=-1\, | \, g \in {\rm SU}(2) \big\} \subset {\rm SU}(2) \approx S^{3}, \]
i.e., the \emph{big equator two-sphere}. Although for such maps the spectral curve does not exist,
they can be found explicitly~\cite{bp}.

\begin{rem}
\emph{Conformal} harmonic maps are the maps $X$ for which the pull-back $X^*(g)$ of the target space metric $g$ is conformally equivalent to the worldsheet metric $h$ on $\Sigma$. They are of special interest in geometry since their images are immersed {\it minimal} surfaces. A slight generalization of conformal harmonic maps are {\it branched conformal} harmonic maps, where $X^*(g)$ becomes degenerate at a finite number of points.
\end{rem}

\begin{rem}
We thank N.~Hitchin for bringing to our attention the reference~\cite{bfpp} where non-conformal harmonic maps~$T^{2} \longrightarrow S^n$ for arbitrary $n$ were constructed from finite-dimensional orbits of symmetries of zero-curvature equations, while in~\cite{burst}, using the ideas from twistor theory and the theory of solitons it was shown that \emph{all} harmonic maps, except for \emph{one special class}, can be found in terms of some finite-dimensional orbits of commuting Hamiltonian flows. The \emph{special class} harmonic maps, the ones outside the reach of the zero-curvature representation framework, were called \emph{pseudoholomorphic} in~\cite{calabi}, \emph{super-minimal} in~\cite{bryant} or \emph{isotropic} in~\cite{ew}. We treat these maps in the second part of the paper.
\end{rem}

\subsection{Spheres from Euclidean spaces, and complexification}

Let ${\BC}^{N}$ be equipped with a non-degenerate quadratic form ${\bf q} \mapsto ({\bf q},{\bf q})$.
The group $G_{\BC}$ of linear transformations ${\bf q} \mapsto g \cdot {\bf q}$ preserving $({\bf q}, {\bf q})$
is isomorphic to $O(N, {\BC})$. The hypersurface $M_{\BC}$ defined by $({\bf q},{\bf q}) = 1$
has the complex metric $g_{\BC}$ induced from the complexified Euclidean metric $({\rm d}{\bf q}, {\rm d}{\bf q})$ on ${\BC}^{N}$. In the basis ${\bf e}_{i}$, $i = 1, \dots , N$ in ${\BC}^{N}$ in which $({\bf q},{\bf q})$ looks like
\[
({\bf q},{\bf q}) = \sum_{i=1}^{N} \big(q^{i}\big)^{2}
\]
the real slice $M = S^{N-1}$ corresponds to $q^{i} \in {\BR}$. There are other real slices, of de Sitter, anti-de Sitter, or Lobachevsky geometry, where some of the $q^i$'s are purely imaginary.

Turning on the twists breaks the $O(N, {\BC})$ symmetry down to its maximal torus. Given a flat $O(N, {\BC})$-connection $A_{\BC}$ on $\Sigma = S^1 \times S^1$ one can find a basis ${\bf e}_{i}$ such that
the holonomies
\begin{gather*}
g_{x} = P \exp \left(\oint {\rm d}x ( A_{\BC})_{x}\right) , \qquad g_{y} = P \exp \left( \oint {\rm d}y (A_{\BC})_{y}\right)
%\label{eq:holon}
\end{gather*}
(the flatness guarantees the conjugacy classes of the path ordered exponentials do not depend on~$y$ and~$x$, respectively)
are diagonal: for $N = 2n$,
\begin{gather}
\begin{split}
& g_{x} ( {\bf e}_{2a-1} \pm {\ii} {\bf e}_{2a} ) = w_{x, a}^{\pm 1} ( {\bf e}_{2a-1} \pm {\ii} {\bf e}_{2a} ) , \\
& g_{y} ( {\bf e}_{2a-1} \pm {\ii} {\bf e}_{2a} ) = w_{y, a}^{\pm 1} ( {\bf e}_{2a-1} \pm {\ii} {\bf e}_{2a} ) , \qquad a = 1, \dots , n,
\end{split}
\label{eq:eventwists}
\end{gather}
with some complex numbers $w_{x,a}, w_{y, a} \in {\BC}^{\times}$, $a = 1, \dots, n$, while for $N = 2n+1$~\eqref{eq:eventwists} are ac\-companied by additional four choices:
\[
g_{x} {\bf e}_{N} = \pm {\bf e}_{N} , \qquad g_{y} {\bf e}_{N} = \pm {\bf e}_{N} .
\]

\begin{rem}
Another deformation of the sigma-model breaking $O(N, {\BC})$ invariance is to disentangle the quadratic forms
$({\rm d}{\bf q}, {\rm d}{\bf q})$ and $({\bf q},{\bf q})$. In this way on gets an ellipsoid (in modern parlance, a squashed sphere)
sigma model. Its point particle limit, i.e., the study of geodesics on the ellipsoid, is connected~\cite{Mumford} to the so-called winding string \cite{KN} case of the twisted sigma model on a sphere. The connection in the general case is not known.
\end{rem}

In the spherical $O(N)$-symmetric case, the embedding into the vector space allows to reformulate the sigma model as a constrained linear sigma model:
\begin{gather}
S \sim \int_{\Sigma} ( {\pa}_{z} {\bf q}, {\pa}_{\zb} {\bf q} ) + u ( ({\bf q},{\bf q} ) - 1 ),
\label{eq:linear}
\end{gather}
where $u = u(x,y)$ is a double-periodic function on $\Sigma$, the Lagrange multiplier.

\subsection{Fermi-curves}
Our approach to solving the $O(N)$ sigma model equations does not use their zero-curvature representation. It is a further development of the scheme proposed in \cite{kr87,kr94} and extended in~\cite{KN}. Roughly, this approach can be described as the construction of {\it integrable linear operators with self-consistent potentials}. The construction consists of two steps. The first step is to parametrize a periodic linear operator by \emph{a~spectral curve and a line bundle $($divisor$)$ on it}. The spectral curve parametrizes Bloch solutions of the linear equation. The second step of the construction is the characterization of the spectral curves for which there exists a set of points on the curve such that the corresponding Bloch solutions satisfy quadratic relation.

Unlike the zero-curvature representation,
the linear equations
\begin{gather}\label{shrod}
 ( -{\pa}_z{\pa}_{\zb}+u ) q^{i} (z,\zb) =0 , \qquad i=1,\dots,N,
\end{gather}
are the equations of motion following from \eqref{eq:linear} and not some auxiliary linear equation. The potential $u$ is the Lagrange multiplier enforcing the constraint that the vector ${\bf q} = (q^{i})_{i=1}^{N} \in {\BC}^{N}$ with the coordinates $q^i$ lies on the complexified unit sphere, i.e.,
\begin{gather}\label{const}
 ( {\bf q} , {\bf q} ) = 1.
\end{gather}
From the equations~\eqref{shrod} and~\eqref{const} it is easy to express $u$ in terms of the solutions to the linear equation:
\begin{gather}\label{u}
u = - ( {\pa}_z {\bf q} , {\p}_{\zb} {\bf q} ).
\end{gather}

\subsection{Novikov--Veselov hierarchy}
In the next section we show that the Novikov--Veselov (NV) hierarchy is a symmetry of the~$O(N)$ model.
Recall, that NV hierarchy \cite{nv1,nv2} is the system of compatible linear equations
\begin{gather}\label{NVa}
H\psi:=(-\p_z\p_{\zb}+u)\psi=0,
\\
\label{NVb}
(\p_{t_n}-L_n)\psi=0 ,
\end{gather}
with
\begin{gather*}%\label{Ln}
L_{n} = {\rm L}_{n} ( z, {\zb}, {\pa}_{z}) = {\p}_z^{2n+1}+\sum_{i=1}^{2n-1}w_{i,n}(z,\zb)\p_z^i ,
\end{gather*}
 differential operator in the variable $z$ with $(z,{\zb})$-dependent coefficients.
The compatibility condition of the equations~(\ref{NVa}), (\ref{NVb}) is the so-called Manakov's triple equation
\begin{gather}\label{manak}
{\pa}_{t_{n}} H = [L_{n},H]+B_{n} H,
\end{gather}
where $B_n= {\rm B}_{n}(z, {\zb}, {\p}_{z})$ also a differential operator in the variable $z$ with $(z, {\zb})$-dependent coefficients.

\begin{rem} The complete set of times in the NV hierarchy contains twice as many variables:
$(t,\bar t):=(t_1,\bar t_1,t_2,\bar t_2,\dots)$. The equation defining the evolution under the times $\bar t_n$ have the Manakov triple form with $( L_n,B_n )$ replaced by the differential operators
\[ {\bar L}_n = {\bar{\rm L}}_{n}(z, {\zb}, {\pa}_{\zb}), \qquad {\bar B}_n = {\bar{\rm B}}_{n}(z, {\zb}, {\pa}_{\zb}) \]
in the variable~$\zb$.
\end{rem}

To prove that the
NV hierarchy is the hierarchy of symmetries of the $O(N)$ sigma model
means to prove the constraint~\eqref{const} remains invariant along the $t_n$ and ${\bar t}_m$ evolution,
given by
\begin{gather*}%\label{q_dyn}
\p_{t_n}q^i = L_n q^i , \qquad \p_{{\bar t}_{m}} q^{i} = {\bar L}_{m} q^{i} .
\end{gather*}
As we shall see soon, this
is highly nontrivial, yet purely local statement, i.e., not relying on any assumptions on the global behaviour of solutions.

The Manakov triple-type equations assume the Lax form upon a restriction to the space of zero modes of the operator $H$, i.e., onto the space of solutions of equations~\eqref{NVa}. Roughly speaking, they preserve the \emph{spectral data} associated with \emph{zero-eigenlevel} of the operator $H$. For operators $H$ with doubly periodic potential~$u$, the latter means that the flows defined by the equation~\eqref{manak} preserve the Bloch--Floquet set $C_u$ of the operator $H$: the set of pairs $p:=(w_{x},w_{y})\subset {\BC}^{\times} \times {\BC}^{\times}$ for which there exists a double Bloch solution
${\psi}_{p} (z, {\zb} ) = {\psi} ( z,{\zb}; p )$ of the equation~\eqref{NVa}:
\begin{gather}\label{bloch}
T_{\a}^{*} {\psi}_{p} = w_{\a} (p) {\psi}_{p} ,\qquad {\alpha} = x, y .
\end{gather}
In \cite{kr-spec} it was shown that for a generic smooth periodic potential the locus $C_{u}$ is a smooth Riemann surface of infinite genus. Moreover, it was shown that the \emph{algebraically-integrable potentials} are dense in the space of all smooth periodic potentials. The latter are the potentials for which the normalization ${\C}_{u}$ of $C_{u}$ called the \emph{Fermi-curve} is of finite genus. For such potentials~${\C}_{u}$ is compactified by two smooth infinity points $P_{\pm}$.

Another characterization of algebraically integrable potentials similar to the one that defines the famous finite-gap potentials of the one-dimensional Schr\"odinger operator is as follows: they are such potentials, that among the flows given by the equations~\eqref{manak}, only a {\it finite number} are linearly independent.

The $t_n$ derivative of \eqref{shrod} with $u$ given by \eqref{u} gives that the vector ${\bf v}_{n} = L_{n} {\bf q}$ satisfies the linearized equation:
\begin{gather}
{\mathfrak{D}} {\bf v}_{n} = 0 ,\nonumber \\
{\mathfrak{D}} =
\p_z\p_{\zb}+ \big( {\bf q}\otimes(\p_{\zb} {\bf q})^t \big)\p_z+ \big( {\bf q} \otimes(\p_{z} {\bf q})^t \big)\p_{\zb}+( {\pa}_z {\bf q},{\pa}_{\zb}{\bf q}) .\label{var}
\end{gather}
The elliptic operator \eqref{var} is defined on $T^{2}$, i.e., its coefficients are double periodic. Therefore $\operatorname{ker}({\mathfrak{D}})$ is finite dimensional. Hence, for any solution of the elliptic double periodic $O(N)$ sigma model the corresponding potential $u$ \eqref{u} is algebraically integrable, i.e., the associated Fermi curve~$\C_u$ is of finite genus.

\begin{rem} That result is an analog of the above mentioned key result in~\cite{hitchin90}.
\end{rem}

Since it was proved in \cite {KN} that a smooth {\it irreducible} Fermi curve of finite genus gives a solution to the periodic $ O (N) $ sigma model if and only if there is a certain meromorphic function on it, it is tempting to say that a periodic $ O(N)$ sigma model is {\it algebraically integrable}.
Unfortunately, this is only partially true. The periodicity constraints that such a curve should satisfy are
given in terms of periods of meromorphic differentials on the curve. They are transcendental and hard to control.

In Section~\ref{s:Winfty} we extend our construction to the case of {\it reducible} Fermi curves. It gives
branched conformal maps of two-torus to the real {\it even} dimensional spheres $M=S^{2n}$, i.e., the periodic solutions of the two-dimensional $O(2n+1)$ sigma model. It turns out that for the reducible Fermi curve the periodicity constraints admit an explicit solution: the irreducible components are the spectral curves of the elliptic Calogero--Moser system corresponding to its {\it turning points}. We recall these notions in what follows.

\begin{rem} It should be emphasized that the conformal maps to $S^2$ are given by our construction with {\it nontrivial} reducible Fermi curves. As it was mentioned above for such maps the spectral curves arising in the framework of the zero-curvature representation are trivial. Therefore, in general these two types of the spectral curves are different. The question: are they interconnected in some special cases seems very interesting but at the moment is open.
\end{rem}

\section{Novikov--Veselov hierarchy}\label{s:NV}
The NV hierarchy in its original form \eqref{manak} is a system of equations on the potential $u$ of the Schr\"odinger operator $H$ and coefficients of operators $L_n$. These are well-defined equations in the sense that the number of equations is equal to the number of unknowns. Recently, their Sato-type representation was proposed in \cite {carpentier}. Below, following
the work of \cite{prym}, we present the hierarchy
in yet another form, the most suitable for our purposes. Namely, as a system of commuting flows on the space of special wave solutions of the Schrödinger equation.

\begin{rem} Throughout this section $z$ and $\zb$ are independent variables.
\end{rem}

\subsection{The phase space}
First recall that a wave (a.k.a.\ formal Baker--Akhiezer) solution of the Schr\"odinger equa\-tion~\eqref{NVa}$\!$ is a formal solution to~\eqref{NVa} of the form
\begin{gather}\label{psi}
{\psi}( z,{\zb}; k) = {\rm e}^{kz} \left(1+\sum_{s=1}^\infty {\xi}_s(z,\zb)k^{-s}\right).
\end{gather}
Substituting \eqref{psi} into \eqref{NVa} gives a system of equations
\begin{gather}\label{xis}
{\p}_{\zb}\xi_{s+1} = H {\xi}_{s} \equiv - {\pa}_{z}{\pa}_{\zb} \xi_s + u \xi_s , \qquad s=0,1,\dots,
\end{gather}
which recurrently determine all the $\xi_s$ uniquely once the initial conditions $\chi_s(z)=\xi_s(z,0)$ are given.

Thus the space ${\bf P}$ of pairs $(H , {\psi})$, where $H = - {\pa}_{z}{\pa}_{\zb} + u$ and ${\psi}$ is a formal wave solution is identified with the space:
\begin{gather}\label{p0}
{\bf P} \simeq \{ ( u(z,{\zb}) , {\chi}_{1}(z) , {\chi}_{2}(z) , \dots , {\chi}_{n}(z) , \dots ) \},
\end{gather}
i.e., with the space of {\emph{ a function of two variables, and a sequence of functions of one variable}}.

For any formal series $\psi$ of the form \eqref{psi} define the {\it dual} formal series
\begin{gather}\label{psidual}
\psi^*(z,\zb,k)={\rm e}^{-kz}\left(1+\sum_{s=1}^\infty \xi_s^*(z,\zb)k^{-s}\right)
\end{gather}
via the equations
\begin{gather}\label{dualdef}
\res_{\infty} (\psi^*\p_z^s\psi )\frac {{\rm d}k}k=-\delta_{s,0},\qquad s=0,1,\dots.
\end{gather}
Substitution of \eqref{psi} and \eqref{psidual} into the equation~\eqref{dualdef} gives the equations
recurrently defining~$\xi_s^*$ in terms of the coefficients of $\psi$:
\begin{gather}\label{s1}
\xi_1^*+\xi_1=0,\qquad \xi_2^*+\xi_2+\xi_1\xi_1^*+2 {\pa}_{z} \xi_{1}=0, \qquad \dots.
\end{gather}

\begin{lem}\label{lm:dualsol} If $\psi$ is a formal BA solution of the Schr\"odinger equation \eqref{NVa} then the dual BA function solves the same equation
\begin{gather*}%\label{NVaa}
H {\psi}^{*}=0.
\end{gather*}
\end{lem}
\begin{proof} Taking the $\zb$ derivative of the equation~\eqref{dualdef} and using the equation~\eqref{NVa} we get
\begin{gather}\label{apr8}
\res_{\infty} \big( ( {\p}_{\zb}{\psi}^{*} ) {\pa}_z^s \psi \big)\frac{{\rm d}k}{k} = {\pa}_{z}^{s-1}u , \qquad s>0 .
\end{gather}
Then taking the $z$ derivative of the equation~\eqref{apr8} we get
\begin{gather}\label{apr81}
\res_{\infty} \big( ( {\pa}_{z} {\p}_{\zb} {\psi}^{*} ) {\pa}_{z}^s \psi \big)\frac{{\rm d}k}{k}=0 , \qquad s>0 .
\end{gather}
The homogeneous equations \eqref{apr81} imply that ${\pa}_{z}{\p}_{\zb} \psi^*$ is equal to $\psi^*$ up to a $k$-independent factor, i.e., ${\pa}_{z}{\p}_{\zb} \psi^{*} = {\tilde u} (z,\zb)\psi^*$. Comparing the leading terms and using the first equations in~\eqref{xis} and~\eqref{dualdef} we get ${\tilde u} = -{\p}_{\zb}{\xi}_{1}^{*}={\p}_{\zb}\xi_1=u$. The lemma is proved.
\end{proof}

\begin{dfn} A {\it reflected} formal BA solution of equation (\ref{NVa}) ${\psi}^{\sigma}$ is defined by
\begin{gather*}%\label{duals}
\psi^\s(z,\zb; k):=\psi(z,\zb; -k).
\end{gather*}
\end{dfn}
Lemma~\ref{lm:dualsol} can be thus formulated as follows: the map
\begin{gather*}
s\colon \ {\psi} \to \left( {\psi}^{*} \right)^{\sigma} %\label{eq:smap}
\end{gather*}
is an involution of $\bf P$, $s^{2} = {\rm id}$.

\begin{dfn} A {\it self-dual} formal BA solution of equation (\ref{NVa}) is fixed point of $s$, i.e., a wave solution that satisfies the constraint
\begin{gather}\label{dual6}
\psi^\s(z,\zb; k):= \psi(z,\zb; -k) = \psi^*(z,\zb; k) ,
\end{gather}
equivalently
\begin{gather*}%\label{dual7}
\xi_s^*=(-1)^s\xi_s
\end{gather*}
for all $s \geq 1$.
\end{dfn}

Our next goal is to show that
the space $\P$ of pairs $\{ ( H,\psi) \}$, where $H$ is a two-dimensional Schr\"odinger operator and $\psi$ its self-dual wave solution, i.e.,
\begin{gather*}%\label{p}
\P:={\bf P}^{s} = \big\{ H, \psi \,| \,H\psi = 0 ,\, \psi^*=\psi^{\s} \big\}
\end{gather*} can be identified with the space
\begin{gather*}%\label{p1}
\P \simeq \{ ( u(z,{\zb}) , {\chi}_{1}(z) , {\chi}_{3}(z) , \dots , {\chi}_{2n-1}(z) , \dots ) \},
\end{gather*}
where $\chi_{2n+1}(z)$ are the part of the data in~\eqref{p0} with odd indices.

Observe that the first equation in \eqref{s1} is satisfied with $\xi_1^*=-\xi_1$, the second equation determines $\xi_2=\xi_2^*$ in terms of $\xi_1$, then the third equation is automatically satisfied by $\xi_3^*=-\xi_3$. It turns out that the same remains true for all odd~$s$. Namely,
\begin{lem}\label{xiodd} Assume the equations
\begin{gather}\label{dodd}
\res_{\infty}\big(\psi^\s\p_z^s\psi\big)\frac {{\rm d}k}k=-\delta_{s,0},\qquad s=0,1,\dots,2n,
\end{gather}
are satisfied. Then
\begin{gather}\label{dodd1}
\res_{\infty}\big(\psi^\s\p_z^{2n+1}\psi\big)\frac {{\rm d}k}k=0
\end{gather}
holds.
\end{lem}

\begin{proof} Using the identity
\begin{gather}\label{dodd2}
\p_z\big((\p_z^i\psi^\s) \p_z^{j}\psi\big)=
\big((\p_z^{i+1}\psi^\s) \p_z^{j}\psi\big)+\big((\p_z^i\psi^\s) \p_z^{j+1}\psi\big)
\end{gather}
it is easy to show by induction in $i$ that the equations~\eqref{dodd} imply
\begin{gather}\label{dodd3}
\res_{\infty}\big(\big(\p_z^i\psi^\s\big)\p_z^{j}\psi\big)\frac {{\rm d}k}k=0, \qquad 0<i+j\leq 2n.
\end{gather}
From (\ref{dodd3}) with $i=j=n$ and the identity $\res_\infty f(k){\rm d}k=-\res_\infty f(-k){\rm d}k$, we get
\begin{gather}\label{dodd4}
0=\p_z\left(\res_{\infty}\big(\big(\p_z^n\psi^\s\big)\p_z^{n}\psi\big)\frac {{\rm d}k}k\right)=
2 \res_\infty\big(\big(\p_z^n\psi^\s\big)\p_z^{n+1}\psi\big)\frac {{\rm d}k}k.
\end{gather}
Then, by induction in $i$ using (\ref{dodd2}) we get that (\ref{dodd4}) implies (\ref{dodd1}).
The lemma is proved.
\end{proof}

Now we are going to present the constraint \eqref{dual6} defining the self-dual wave solution in another, equivalent, form. For that first recall, that the formal adjoint to $w\cdot {\pa}_{z}^{i}$ is the operator
\[
\big( w\cdot {\pa}_{z}^{i} \big)^{*} = (-{\pa}_{z} )^{i} \cdot w,
\]
where $w$ stands for the operator of multiplication by the function $w(z, {\zb})$. Below we will often use the notion of the left action of an operator which is identical to
the formal adjoint action, i.e., the identity
\begin{gather}\label{ad2}
(f\D):=\D^*f.
\end{gather}

\begin{lem} Let $\Phi$ denote the so-called \emph{wave operator}
\begin{gather*}%\label{S}
\Phi=1+\sum_{s=1}^{\infty}\xi_s(z,\zb)\p_z^{-s}
\end{gather*}
corresponding to the formal BA solution \eqref{psi} of the equation~\eqref{NVa}. Then the constraint~\eqref{dual6} is equivalent to the operator equation
\begin{gather}\label{FF}
\Phi^*=\p_z\cdot \Phi^{-1}\cdot\p_z^{-1},
\end{gather}
where $\Phi^*$ is the formal adjoint to $\Phi$,
\[
{\Phi}^{*} = 1 + \sum_{s=1}^{\infty} \big( {-}{\pa}_{z})^{-s} \cdot {\xi}_{s}(z, {\zb}\big) .
\]
\end{lem}

\begin{proof} By definition of $\Phi$ we have
\begin{gather*}%\label{apr83}
\psi=\Phi {\rm e}^{kz}, \qquad \psi^\s=\Phi {\rm e}^{-kz}.
\end{gather*}
Therefore, for the proof of \eqref{FF} it is enough to show that for the function
\begin{gather*}%\label{psidual1}
\psi^{*} = \big( {\rm e}^{-kz}\p_z\cdot \Phi^{-1}\cdot {\p}^{-1}_z \big)
\end{gather*}
the equations~\eqref{dualdef} hold. The latter is an easy corollary of the identity
\begin{gather}\label{dic}
\res_\infty \big({\rm e}^{-kz}\D_1\big)\big(\D_2 {\rm e}^{kz}\big){\rm d}k=-\res_{\p} (\D_2\D_1 )
\end{gather}
valid for any pseudo-differential operators $\D_1$ and $\D_2$. Indeed,
\begin{gather*}%\label{dic1}
\res_\infty\big({\rm e}^{-kz}\p_z\cdot \Phi^{-1}\cdot \p_z^{-1}\big)\big(\p_z^s\cdot\Phi {\rm e}^{kz}\big)\frac{{\rm d}k}k=-\res_{\p} \big(\p_z^s\cdot\Phi\cdot\Phi^{-1}\cdot\p_z^{-1}\big)=-\res_{\p} \p_z^{s-1}.
\end{gather*}
The lemma is proved.
\end{proof}

\begin{cor} \label{bkp} Let $\psi$ be a self-dual formal BA function then the pseudo-differential operator
\begin{gather*}%\label{LL}
{\cL} = {\pa}_{z} + \sum_{s=1}^{\infty} v_s(z,\zb) {\pa}_z^{-s}
\end{gather*}
such that the equation
\begin{gather}\label{kk}
{\cL} {\psi} = k \psi
\end{gather}
holds satisfy the equation
\begin{gather}\label{ad1}
{\cL}^{*} = -{\pa}_z {\cL} {\pa}_z^{-1}.
\end{gather}
\end{cor}
To prove \eqref{ad1} note, that ${\cL} = {\Phi}\cdot {\pa}_z \cdot \Phi^{-1}$, then use~\eqref{FF}.

\begin{rem} The corollary \ref{bkp} identifies the space of operators $\cL$ corresponding to the
self-dual formal BA solution of the Schr\"odinger equation (at fixed~$\zb$) with the phase space of the so-called BKP hierarchy.
\end{rem}

\subsection{The flows}

Let us denote the \emph{strictly} positive differential part of the pseudo-differential operator ${\cL}^n$ by
${\cL}^n_{+}$, i.e., if
\[
{\cL}^n=\sum_{i=-n}^{\infty} F_{n}^{(i)}\p_z^{-i},
\]
then
\begin{gather*}%\label{ad5}
{\cL}^{n}_{+} = \sum_{i=1}^{n} F_{n}^{(-i)}\p_z^{i},\qquad
{\cL}^{n}_{-} = {\cL}^{n}-{\cL}^{n}_{+} = F_{n}^{(0)}+F_{n}^{(1)}{\pa}_{z}^{-1}+O\big({\pa}_{z}^{-2}\big).
\end{gather*}
\begin{rem}
Note that this definition differs from the one used in the KP theory. There
the~``$^{+}$'' subscript denotes the non-negative part of a pseudo-differential operator.
\end{rem}

By definition of the residue of a pseudo-differential operator, the first leading coefficients
of~${\cL}^{n}_{-}$ are
\begin{gather}\label{res1}
F_n^{(0)}={\res}_{ \pa} \big( {\cL}^{n}{\pa}_{z}^{-1}\big),\qquad F_{n}^{(1)}={\res}_{ \p}\ {\cL}^n.
\end{gather}
From the equation~\eqref{ad1} and the relation
${\res}_{\p} D=-{\res}_{\p} D^*$ it follows that
\begin{align}
F_{2n+1}^{(0)}&=-\operatorname{res}_{ \p}\big( {\cL}^{2n+1}{\pa}_{z}^{-1}\big)^*=
\operatorname{res}_{ \p} \big(\p_z^{-1}\cdot\big( {\cL}^{2n+1}\big)^*\big)\nonumber \\
 &= -\operatorname{res}_{ \pa} \big( {\cL}^{2n+1} {\pa}_{z}^{-1}\big)=
-F_{2n+1}^{(0)}=0.\label{sb5}
\end{align}
The latter implies
\begin{gather}\label{adjL}
\big( {\cL}^{2n+1}_{+} \big)^{*} = {\pa}_{z} \cdot {\cL}^{2n+1}_+ \cdot {\pa}_{z}^{-1}.
\end{gather}

In what follows we need the following statement (see \cite[Lemma~4.3]{prym}).
\begin{lem} The equation
\begin{gather}\label{sb3}
H{\cL}^{2n+1}_{+} = - F_{2n+1,\zb}^{(1)}+B_{2n+1}H,
\end{gather}
where $B_{2n+1}$ is a pseudo-differential operator in the variable $z$, holds.
\end{lem}
\begin{proof} Each operator $\D$ of the form $\D=\sum_{i=N}^{\infty} (a+b\p_{\zb})\p_z^{-i}$
can be uniquely represented in the form $\D = D_1+D_2H$,
where $D_{1,2}$ are pseudo-differential operators in the variable $z$.
Consider such a representation for the operator $H{\cL}^m=D_1+D_2H$. From the definition
of $\cL$ it follows that $H{\cL}^m\psi=0$. That implies $D_1=0$ or the equation
\begin{gather}\label{ad12}
H{\cL}^m=D_2H.
\end{gather}
We have the identity
\begin{gather*}%\label{ad9}
\big[ \p_{\zb}\p_z-u,{\cL}^n_{+} \big] = {\cL}_{+, z,\zb}^n+ {\cL}^n_{+,\zb}\p_z-\big[u,{\cL}^n_+\big]+{\cL}^n_{+, z}\p_z^{-1}\cdot u
-{\cL}^n_{+, z}\cdot\p_z^{-1}\cdot H .
\end{gather*}
The first three terms are differential operators in the~$z$ variable. By definition of
${\cL}^n_{+}$ the fourth term is also a differential operator.
Therefore, the pseudo-differential operator~$D_{n,1}$ in the decomposition
$H{\cL}^n_{+} = D_{n,1}+B_{n} H$ is in fact a \emph{differential} operator.

In the same way we get the equation
\begin{gather*}%\label{ad13}
H{\cL}^{n}_{-} = {\tilde D}_{n,1}+ {\tilde B}_{2n} H ,
\end{gather*}
where
\begin{gather}\label{ad14}
\tilde D_{n,1} = {\cL}_{-, z,\zb}^n+{\cL}^n_{-,\zb}\p_z-\big[u,{\cL}^n_-\big]+{\cL}^n_{-, z}\p_z^{-1}\cdot u.
\end{gather}
From the equation~\eqref{sb5} it follows that the operator $\tilde D_{2n+1,1}$ is a pseudo-differential
operator of order not greater than $0$.
The equation~\eqref{ad12} implies $H{\cL}_+^n=-H{\cL}^n_-+D_2H$. Hence,
$D_{2n+1,1}=-\tilde D_{2n+1,1}$ is a differential operator of the order~$0$, i.e., it is an operator of multiplication by a function. This function is easily found from
the leading coefficient of the right-hand side of the equation~\eqref{ad14}.
Direct computations give the equation~\eqref{sb3}. The lemma is proved.
\end{proof}

Now we are ready to define NV hierarchy explicitly.
\begin{theo} \label{thm:man} The equations
\begin{gather}\label{manu}
{\pa}_{t_{n}} u = {\pa}_{\zb}F_{2n+1}^{(1)} ,
\\
\label{manpsi}
{\pa}_{t_{n}} \psi=\big( {{\cL}}^{2n+1}_{+} - k^{2n+1} \big) \psi
\end{gather}
define a family of commuting flows on the space $\P$ of self-dual formal BA solutions of the
Schr\"odinger operators.
\end{theo}

\begin{proof} From (\ref{kk}) and definition of the operator $ {\cL}^{2n+1}_+$ it follows that
\begin{gather*}%\label{eqpsi}
\big({\cL}^{2n+1}_+-k^{2n+1}\big) \psi=-{\cL}^{2n+1}_- \psi.
\end{gather*}
The equation~\eqref{sb5} implies that the operator ${\cL}^{2n+1}_-$ is of order at most~$-1$. Hence,
the equation~\eqref{manpsi} is equivalent to a well-defined system of equations on the coefficients~$ \xi_s $ of the formal BA function $\psi$.

In turn, the equation~\eqref{sb3} implies that the equation~\eqref{manu} is equivalent to the operator equation~\eqref{manak} with $L_n={\cL}_{+}^{2n+1}$, which is a compatibility condition for the Schr\"odinger equation and the equation~\eqref{manpsi}. Hence, the vector field defined by the right-hand sides of (\ref{manu}) and~(\ref{manpsi}) is tangent to the space of formal Baker--Akhiezer solutions of the Schr\"odinger equation. Let us show that it is tangent to the space~$\P$ of selfdual solutions, i.e., preserve constraint~\eqref{dual6}.

From the equation~\eqref{manpsi} and the same equation with $k$ replaced by $-k$ it follows that
\begin{gather}\label{tanp}
\p_{t_n} \big( {\psi}^{\s} \big( {\pa}_{z}^{s} {\psi} \big) \big) = {\psi}^{\s} \big( {\pa}_{z}^{s}\cdot {\cL}_{+}^{2n+1}{\psi} \big) + \big( \big( {\cL}_{+}^{2n+1}{\psi}^{\s} \big){\pa}_{z}^{s} {\psi} \big).
\end{gather}
It is easy to show that in conjunction with the self-duality $\psi^*=\psi^\s$ the equation~\eqref{dualdef} implies
\begin{gather*}%\label{sdual}
{\res}_{\infty} \left( \big( {\pa}_z^{i} {\psi}^{\s} \big) \big( {\pa}_{z}^{j} {\psi} \big) \frac{{\rm d}k}{k} \right) = 0 , \qquad i+j>0.
\end{gather*}
The coefficients of the operator ${\cL}_+^{2n+1}$ are $k$-independent. Therefore, from the equations~\eqref{tanp} and \eqref{sb5} we deduce
\begin{gather*}%\label{sdualinv}
\p_{t_n} {\res}_{\infty} \left( \psi^{\s} \big( {\pa}_{z}^{s} {\psi} \big) \frac{{\rm d}k}{k} \right)=0 , \qquad s \geq 0 ,
\end{gather*}
i.e., that $\p_{t_n}$ is tangent to $\P$.

The commutativity of NV flows is easy to demonstrate. Indeed, \eqref{manpsi} implies
\begin{gather}\label{comf}
\big[ {\pa}_{t_n}-{\cL}^{2n+1}_+, {\pa}_{t_{m}}-{\cL}^{2m+1}_+ \big] \psi=0.
\end{gather}
The operator in the left-hand side of the equation~\eqref{comf} is a differential operator in the variable~$z$. From the very form of $\psi$ it follows that if such an operator annihilates it, then it itself must be zero. The theorem is proved.
\end{proof}

\subsection{Integrals}

A direct corollary of the Schr\"odinger equation for a self-dual BA function is the equation
\begin{gather*}%\label{I}
{\pa}_{\zb}\big({\psi}^{\s} {\psi}_{z}-{\psi}^{\s}_{z}{\psi} \big)+
{\pa}_{z}\big( {\psi}^{\s} {\psi}_{\zb}-{\psi}^{\s}_{\zb}{\psi} \big)=0,
\end{gather*}
which is nothing but the conservation
\begin{gather}
{\rm d}_{\Sigma} * j = 0
\label{eq:jconservation}
\end{gather}
of the current
\begin{gather}
j = {\psi}^{\s} {\rm d}_{\Sigma} {\psi} - {\psi} {\rm d}_{\Sigma} {\psi}^{\s} .
\label{eq:currento}
\end{gather}
In other words the $1$-form $j_{z} {\rm d}z-j_{\zb} {\rm d}\zb$, where
\begin{gather}\label{J}
j_{z}: ={\psi}^{\s} {\pa}_{z}\psi- ( {\pa}_{z} \psi^{\s} ) \psi = 2\left(k+\sum_{n=0}^{\infty} J_n k^{-2n-1}\right), \\
j_{\zb} := {\psi}^{\s} {\pa}_{\zb} \psi - ({\pa}_{\zb} {\psi}^{\s} ) \psi = 2\sum_{n=0}^{\infty} {\bar J}_n k^{-2n-1}
\label{J1}
\end{gather}
is closed.
The current \eqref{eq:currento} is the Noether current corresponding to the maximal commutative
subalgebra of an infinite orthogonal group, acting on the formal BA functions by linear transformations,
preserving the $\operatorname{res}_{\infty} {\rm d}k/k \psi^{\s}{\psi}$ pairing, namely,
\[
J(k_1, k_2) = {\psi}^{\s}(k_{1}) {\rm d}_{\Sigma} {\psi}(k_{2}) - {\rm d}_{\Sigma} {\psi}^{\s}(k_{1} ) {\psi}(k_2) = - J(-k_2, -k_1)
\]
obeys
\[
{\rm d}_{\Sigma} \star J (k_1, k_2) = 0 , \qquad {\rm d}_{\Sigma} J (k_1, k_2 ) + \operatorname{res}_{\infty} \frac{{\rm d}k}{k} J (k_1, k) \wedge J (k, k_2) = 0.
\]
By definition \eqref{ad2} of the left action of a differential operator $\D$ we have the identity
\begin{gather*}%\label{ad25}
(h \D)f=h(\D f)+\sum_{i=1}^{d-i} \p_z^i\big(h\big(\D^{(i)} f\big)\big),
\end{gather*}
where $d$ is the order of $\D$ and $\D^{(i)}$ is differential operator of degree $d-i$ whose coefficients are explicit differential polynomials in the coefficients of the operator~$\D$.
Therefore,
\begin{gather}\label{JQ1}
j_{z}=\big({\rm e}^{-kz}\p_z \Phi^{-1}\p_z^{-1}\big)\big(\p_z\Phi {\rm e}^{kz}\big)
+\big({\rm e}^{-kz}\p_z \Phi^{-1}\big) \big(\Phi {\rm e}^{kz}\big)=2k+\p_z Q^{(1)},
\end{gather}
where the coefficients of the series
\begin{gather*}%\label{QQ}
Q^{(1)}=\sum_{s=0}^\infty Q^{(1)}_n k^{-2s-1}
\end{gather*}
are universal polynomials in the coefficients of the self-dual BA function.

For further use, let us show that the coefficient $J_n$ of the series \eqref{J} is equal to the function~$F_{2n+1}^{(1)}$ in the right-hand side of equation~\eqref{manu}. Indeed, from \eqref{FF}, \eqref{kk}, \eqref{dic} and~\eqref{sb5} it follows that
\begin{align}
2J_n&= \operatorname{res}_k \big(\big(\psi^\s{\cL}^{2n+1}\big)\psi_z-\psi_z^\s\big({\cL}^{2n+1}\psi\big)\big)k^{-1}{\rm d}k \nonumber\\
&= \operatorname{res}_{\p} \big({\cL}^n+\p_z{\cL}^n\p_z^{-1}\big)=2F_{2n+1}^{(1)}.\label{sb6}
\end{align}
 We conclude this section by proving that $J_n$ are densities of the NV hierarchy integrals.

\begin{lem} The equations
\begin{gather*}%\label{int}
\p_{t_{n}} J_n=\p_z Q_n,
\end{gather*}
where $Q_n$ are explicit differential polynomials in the coefficients $\xi_s$ of the self-dual BA solution of the Schr\"odinger equation, hold.
\end{lem}

\begin{proof}
The equation~\eqref{manpsi} and the same equation for $\psi^\s$ with $k$ replaced by $-k$ imply
\begin{gather*}%\label{I11}
\p_{t_{n}} J = \psi^\s\big(\p_z\big({\cL}_+^{2n+1}\psi\big)\big)+
\big(\psi^\s\big({\cL}_+^{2n+1}\big)^*\big)(\p_z\psi) \\
\hphantom{\p_{t_{n}} J =}{}
+ \big(\psi^\s\big({\cL}_+^{2n+1}\big)^*\p_z\big)\psi+\big(\psi^\s \p_z\big)\big({\cL}^{2n+1}\psi\big).
\end{gather*}
Therefore, from \eqref{sb6} and \eqref{adjL} it follows that $\p_{t_n} J=\p_z Q_{n}$ where the coefficients of~$Q_{n}$ are explicit polynomials in the coefficients $\xi_s$ of the self-dual BA solution~$\psi$. The lemma is proved.
\end{proof}

\section[The NV flows as symmetries of O(N) model]{The NV flows as symmetries of $\boldsymbol{O(N)}$ model}

Our next goal is to show that the NV hierarchy defines a set of commuting symmetries of the~$O(N)$ sigma model.

First observe that if the equations
\begin{gather}\label{wn}
\big(\p_{z}^j {\bf q},\p_{z}^j {\bf q}\big)=0,\qquad 0<j<m,
\end{gather}
hold then \eqref{shrod} implies ${\p}_{\zb}\big(\p_z^m {\bf q},\p_z^m {\bf q}\big)=0$.

We will call the solutions of the $O(N)$ model satisfying \eqref{wn} \emph{$w_m$-harmonic}. Notice, that for $m>1$ they are conformal. In this way the condition of being $w_{m}$-harmonic defines a filtration on the space of conformal maps.

The equations of motion (\ref{shrod}), (\ref{u}) of the model are invariant under a conformal change of variables $z \mapsto f(z), {\zb} \mapsto {\bar f} ( {\zb} )$. Hence if for a $w_m$-harmonic map $\big(\p_z^m {\bf q}$, $\p_z^m {\bf q}\big)\neq 0$, then without loss of generality we may assume that the equation
\begin{gather}\label{dq}
\big( {\p}_z^m {\bf q}, {\p}_z^m {\bf q} \big) = (-1)^{m+1}
\end{gather}
holds.

\begin{theo} Let ${\bf q}(z,\zb)$ be a complex solution of the $O(N)$ sigma model obeying \eqref{wn}, \eqref{dq}. Then there exists a unique up to multiplication by $(z,\zb)$-independent factor $\rho(k)$, such that
\begin{gather}\label{psitrans}
\psi \longmapsto \psi \rho(k),\qquad \rho(k)=\exp \left(\sum_{s=1}^\infty \rho_s k^{-2s+1}\right)
\end{gather}
is a self-dual formal BA solution $\psi$ of the Schr\"odinger equation with the potential $u$ given by~\eqref{u} such that the constraint~\eqref{const} is invariant under the commuting flows
\begin{gather}\label{qdyn1}
\p_{t_n} {\bf q} = {{\cL}}_+^{2n+1} {\bf q}
\end{gather}
with $\cL$ defined by \eqref{kk}.
\end{theo}

\begin{rem} Note, that the operator $\cL$ defined by \eqref{kk} is invariant under the transforma\-tion~\eqref{psitrans}. Hence, the right-hand side is uniquely defined by $\bf q$.
\end{rem}

\begin{proof} Consider the \emph{moments} defined by the formula
\begin{gather*}%\label{ti}
T_{i} (z, {\zb}) := \big( {\bf q} , {\pa}_z^{i} {\bf q} \big) , \qquad i=0,1,\dots.
\end{gather*}
The original constraint \eqref{const} translates to $T_0=1$. By assumptions \eqref{wn} \eqref{dq}, the next moments equal $T_i=0$, $1<i<2m$, $T_{2m}=-1$.

Taking the $\zb$ derivative of $T_i$ and using the equation~\eqref{shrod} we conclude that the moments satisfy the triangular system of linear equations
\begin{gather}
 {\pa}_{\zb} T_i = {\pa}_z {\pa}_{\zb} T_{i-1} - {\pa}_z \big( {\bf q} , {\pa}_z^{i-2}(u {\bf q}) \big)+\big( {\bf q} , \big[ {\pa}_z^{i-1} , u\big] {\bf q} \big) \nonumber \\
\hphantom{{\pa}_{\zb} T_i}{}
= {\pa}_{z} \left( {\pa}_{\zb}T_{i-1} - \sum_{l=0}^{i-2} {{{i-2}\choose{l}}} T_{l} {\pa}_{z}^{i-l-2} u \right) +
 \sum_{l=0}^{i-2} {{{i-1}\choose{l}}} T_{l} {\pa}_{z}^{i-1-l} u ,
\qquad i>2.\label{tdiff}
\end{gather}
Therefore, the moments $T_i$ are uniquely determined by the initial conditions $T_{i}(z,0)$ for \eqref{tdiff}.

Now, define the Fermi-moments
\[
f^{[m]}_{i} :=\res_\infty \left( {\psi}^{\s} \big( {\pa}_{z}^{i} {\psi} \big) \frac{{\rm d}k}{k^{2m+1}} \right).
\]

\begin{lem} \label{lm:moments} Under the assumptions of the theorem there is a unique up to the transformation~\eqref{psitrans} self-dual BA solution $\psi$ of the Schr\"odinger equation~\eqref{shrod} such that for $i>1$ the equations
\begin{gather}\label{tq}
T_i=f_i^{[m]}
\end{gather}
hold.
\end{lem}

\begin{proof}The derivation of \eqref{tdiff} uses only the fact that the corresponding quantities are given by some bilinear form on the pair of solutions of the Schr\"odinger equation, i.e., the functions~$f_{i}^{[m]} (z,\zb)$ defined in~\eqref{tq} for any $\psi$ satisfy the same equations~\eqref{tdiff}. Therefore the equations~\eqref{tq} hold identically in $(z,\zb)$ if they are satisfied at~$(z,0)$. We now show by induction that the latter condition allows to recover the functions $\chi_{2n-1}(z)$ defining the self-dual BA function.

We know from the equations~\eqref{wn} and \eqref{dq} that the assumption holds for $n = m$.
Now assume the equation~\eqref{tq} holds for all $i\leq 2n-1$ for a self-dual BA function with fixed $( \chi_1,\dots, \chi_{2n-3})$. Then
\begin{gather*}%\label{tq1}
g_{i,j}:=\big( {\pa}_{z}^j {\bf q} , {\pa}_{z}^i {\bf q} \big) - \operatorname{res}_{\infty} \frac{{\rm d}k}{k^{2m+1}} \big( {\pa}_{z}^j{\psi}^{\s} \big) \big( {\pa}_{z}^i \psi \big) = 0 ,\qquad i+j\leq 2n-1
\end{gather*}
and
\begin{gather*}%\label{tq2}
g_{n-i,n+i}:=(-1)^i g_{n,n}.
\end{gather*}
From \eqref{tdiff} with $i=2n$ (and the same equation for $f_{2n}^{[m]}$) it follows that under the induction assumption ${\pa}_{\zb} (g_{0,2n})=0$. Now we are going to show that with a proper choice of ${\pa}_z\chi_{2n-1}$ the equation~\eqref{tq} with $i=2n$ holds, i.e., $g_{0,2n}=0$.

Indeed, the substitution of \eqref{psi} into the definition of $f_{2n}^{[m]}$ gives
\begin{gather*}%\label{f2n}
f_{2n}^{[m]} = 2\xi_{2n} + R ( {\xi}_1,\dots,{\xi}_{2n-2} ),
\end{gather*}
where $R$ is some explicit differential polynomial of its arguments. Recall, that the self-duality of $\psi$ gives the expression of ${\xi}_{2n}$ in terms of ${\xi}_1,\dots, {\xi}_{2n-1}$ of the form{\samepage
\begin{gather*}%\label{xi2n}
0 = 2{\xi}_{2n}-2{\xi}_1{\xi}_{2n-1}+{\pa}_{z}{\xi}_{2n-1}+{\tilde R} ({\xi}_1,\dots,{\xi}_{2n-2} ), \qquad n>1,
\end{gather*}
where ${\tilde R}$ is also a differential polynomial of its arguments.}

Equation \eqref{tq} with $i=2n$ restricted onto $(z,0)$ is equivalent to the first order differential equation for the function ${\chi}_{2n-1}$ which defines it uniquely if the initial condition ${\chi}_{2n-1}(0)$ is given. The ambiguity in the choice of the latter corresponds to the transformation~\eqref{psitrans}.

If equation $g_{0,2n}=0$ is satisfied then using the equation
\begin{gather*}%\label{tq5}
\p_{z} g_{n,n}=2g_{n,n+1}
\end{gather*}
we get that $g_{0,2n+1}=0$. The induction step is completed and the lemma is proved.
\end{proof}

The theorem statement that the constraint \eqref{const} in invariant under the flows \eqref{qdyn1} is equivalent to the equations $\big(q,{\cL}_+^{2n+1} q\big)=0$. The latter is an easy corollary of~\eqref{tq}. Indeed, since the coefficients of ${\cL}^{2n+1}_+$ are $k$ independent equations~\eqref{tq} imply
\begin{gather*}%\label{Lq}
\big(q,{\cL}_+^{2n+1} q\big)=\res_\infty \big(\psi^\s, {\cL}_+^{2n+1} \psi\big) k^{-2m-1} {\rm d}k \\
\hphantom{\big(q,{\cL}_+^{2n+1} q\big)}{}
=\res_\infty \big(\psi^\s, \psi\big) k^{2n-2m} {\rm d}k+
\res_\infty \big(\psi^\s, {\cL}_-^{2n+1} \psi\big) k^{-2m-1} {\rm d}k=0.
\end{gather*}
For the proof of the last equation it is enough to note that the first term has no residue since it is an odd differential. The second term is of order $O\big(k^{-2m-2}\big)$ and hence also has no residue. The theorem is proved.
\end{proof}

\section[O(N) sigma model on a two-torus]{$\boldsymbol{O(N)}$ sigma model on a two-torus}

In the previous section, the Novikov--Veselov hierarchy was defined as a system of commuting flows on the space of self-dual solutions to the Schr\"odinger equations on a plane ${\BR}^{2}$. Under the assumption that the potential is periodic in one of the variables, i.e., is defined on ${\BR}^{1} \times S^{1}$, the hierarchy can be defined as a system of flows on the space of the Schr\"odinger operators themselves.
The idea of the corresponding construction goes back to~\cite{kp1} where the KP hierarchy was defined as a set of commuting flows on the space of functions of two variables. The starting point of the construction is an observation that for periodic potentials a unique formal BA function can be singled out by the condition that it is Bloch, i.e., an eigenvector of the monodromy around~$S^1$. We refer the reader to \cite{prym, kp1} where the case of two independent variables $z$ and $\zb$ was considered.

In this section, we present a necessary modification of this construction to the case where the potential
$u$ is a doubly periodic function,
\[
u = \sum_{(m,n) \in {\BZ}^{2}} {\hat u}_{m,n} {\rm e}^{2\pi \ii ( m x + n y )},
\]
or, in terms of the $z = x + {\tau}y$, ${\zb} = x + {\bar\tau}y$ coordinates:
\begin{gather*}
%\label{uper}
u \left( z+{\omega}_{\a} , {\zb} + {\bar\omega}_{\a} \right) = u(z,\zb),
\end{gather*}
for ${\a} = x, y$, cf.~\eqref{sigma}.

\subsection{The formal Bloch solutions}

\begin{lem} \label{lm:psibloch}Let $u(z,\zb)$ be a double periodic function. Then there is a unique up to the transformation \eqref{psitrans} self-dual BA formal solution of the Schr\"odinger equation \eqref{shrod} of the form
\begin{gather}\label{psibloch}
\psi(z,{\zb},k)= {\rm e}^{kz+{\ell}(k)\zb}\left(1+\sum_{s=1}^\infty \zeta_s(z,\zb)k^{-s}\right)
\end{gather}
with double periodic coefficients $\zeta_s$, i.e.,
\begin{gather}\label{zetas}
\zeta_{s} \big( z+{\omega}_{\a} , {\zb}+{\bar \omega}_{\a} \big) = \zeta_{s}(z,{\zb}) , \qquad {\a} = x, y,
\end{gather}
and where $\ell(k)$ is a formal series,
\begin{gather}\label{ell}
{\ell}(k) = \sum_{s=1}^\infty {\ell}_s k^{-s}.
\end{gather}
\end{lem}

\begin{proof} Substitution of the equation~\eqref{psibloch} into the equation~\eqref{shrod} gives
\begin{gather}\label{eq:zeta}
{\p}_{\zb}\zeta_{s+1}=-\ell_{s+1}+u\zeta_s-{\pa}_z{\p}_{\zb}\zeta_s-\sum_{i=1}^s\ell_i\zeta_{s-i}, \qquad s=0,1,\dots .
\end{gather}
Suppose that $\zeta_j$ and $\ell_j$ with $j\leq s$ are known. The inhomogeneous first order elliptic equation~\eqref{eq:zeta} has a double periodic, i.e., satisfying \eqref{zetas}, solution iff the integral of its right-hand side vanishes, thus determining the constant ${\ell}_{s+1}$:
\[
{\ell}_{s+1} = \int_{\Sigma} {\rm d}x{\rm d}y \left( u {\zeta}_s -\sum_{i=1}^s \ell_{i} \zeta_{s-i} \right).
\]
Once ${\ell}_{s+1}$ is found, the double periodic solution ${\zeta}_{s+1}$ is defined uniquely up to an additive constant. The ambiguity in the definition of $\zeta_{s+1}$ corresponds to the transformation
$\psi\to \psi \tilde \rho(k)$, where
$\tilde \rho(k)$ is a regular series in the variable $k$ with constant coefficients.

Using the fact the both the dual ${\psi}^*$ and reflected ${\psi}^{\s}$ BA functions are solutions of \eqref{shrod} with double periodic coefficients we get the equation
\begin{gather*}%\label{s22}
{\psi}^{\s} = h(k) {\psi}^{*},
\end{gather*}
for some function $h(k)$.
Under the transformation ${\psi}\mapsto {\tilde \rho}(k) {\psi}$ the identifying function $h(k)$ gets transformed to
\begin{gather*}%\label{s23}
{\tilde h}(k) = {\tilde \rho} (-k)^{-1} {\tilde \rho}^{-1}(k) h(k).
 \end{gather*}
 Thus, by taking $\tilde \rho$ to solve the equation $\tilde \rho(k)\tilde \rho(-k)=h(k)$ we can transform
 $\psi$ having the form~\eqref{psibloch} to a self-dual BA function.
By construction it is unique up to the transformation~\eqref{psitrans}.
The lemma is proved.
\end{proof}

\begin{cor} \label{corbloch} The self-dual solution of the Schr\"odinger equation defined in Lemma~{\rm \ref{lm:psibloch}} satisfies the following monodromy properties
\begin{gather*}%\label{monpsi}
\psi(z+\omega_\a,\zb+\bar \omega_\a,k)=w_\a(k)\psi(z,\zb,k)
\end{gather*}
with
\begin{gather*}%\label{wa}
w_x = {\rm e}^{k+{\ell}(k)}, \qquad w_y = {\rm e}^{{\tau}k + {\bar \tau} {\ell}(k)}.
\end{gather*}
\end{cor}

\begin{theo} The equations
\begin{gather}\label{NVnext}
\p_{t_n}u= {\p}_{\zb} F_{2n+1}^{(1)}
\end{gather}
with $F_{2n+1}^{(1)}=\res_\p {\cL}^{2n+1}_+$, where $\cL$ is defined by the Bloch self-dual solution of \eqref{shrod}, are commuting flows on the space of double periodic functions.

Moreover if $u$ is given by~\eqref{u} for some solution of
$O(N)$-sigma model then the constraint~\eqref{const} and hence the equation~\eqref{u} are preserved by the flows
\begin{gather*}%\label{ONflow}
\p_{t_n} q= {\cL}^{2n+1}_{+} q.
\end{gather*}
\end{theo}
\begin{proof} By the Corollary~\ref{corbloch} the coefficients of the operator $\cL$ corresponding to the self-dual BA solution defined in Lemma~\ref{lm:psibloch} are double periodic functions. Hence, the Bloch self-dual BA solution are invariant under the flows~\eqref{manpsi}, and~\eqref{NVnext} is just a well-defined restriction of the flows described by Theorem~\ref{thm:man}.

The last statement of the theorem is a direct corollary of the following lemma.
\begin{lem}\label{lm:permoments} Let $q$ be a solution of the $O(N)$ sigma model on Euclidean $T^2$ such that the equations~\eqref{ell} and \eqref{dq} hold.
Then there exists a unique formal series
\begin{gather*}%\label{E}
E(k)=\sum_{s=m}^\infty e_s k^{-2s},
\end{gather*}
such that equations
\begin{gather}\label{tqper}
T_i=f^{E}_i:=\res_\infty \left( {\psi}^{\s} \big( {\pa}_{z}^{i} {\psi} \big) E(k)\frac {{\rm d}k}k \right),
\end{gather}
where $\psi$ is the self-dual BA solution defined in the Lemma~{\rm \ref{lm:psibloch}},
hold.
\end{lem}

\begin{rem} Before presenting the proof of the lemma it is instructive to compare it with the Lemma \ref{lm:moments} which was used above to single out a special self-dual solution of the Schr\"odinger equation for which equations~\eqref{tq} hold. In the double periodic case that special solution is singled out in Lemma~\ref{psibloch} by the Bloch properties and in the proof of Lemma~\ref{lm:permoments} below we will demonstrate the validity of~\eqref{tqper} for {\it that} solution. In essence we are arriving from two different starting points at the same self-dual BA solution of the Schr\"odinger equation as well as the same equations~\eqref{tq} and \eqref{tqper}. In view of this, additional comments are required to clarify that, at first glance, the latter equations have different forms.

Note that the form of the formal BA function is invariant under a change of the local coordinate $k \mapsto k+O\big(k^{-1}\big)$. In the double periodic case the local coordinate is fixed by the periodicity conditions \eqref{zetas} while in the Lemma \ref{lm:moments} it is fixed by the condition that in this coordinate the series~$E(k)$ equals~$k^{-2m}$.
\end{rem}

\begin{proof}Suppose that \eqref{tqper} is satisfied for $1\leq i \leq 2n-1$ for the Bloch self-dual BA function and for a series $E(k)$ with some known $( e_{m}=1, e_{m+1},\dots , e_{n-1} )$.

If equation~\eqref{tqper} holds for all $i\leq 2n-1$, then
\begin{gather*}%\label{tq1per}
g_{i,j}:=\big( {\pa}_{z}^j {\bf q} , {\pa}_{z}^i {\bf q} \big) - \operatorname{res}_{\infty} \left( \big( {\pa}_{z}^j{\psi}^{\s} \big) \big( {\pa}_{z}^i \psi \big) E(k)\frac {{\rm d}k} k \right) = 0 ,\qquad i+j\leq 2n-1,
\end{gather*}
and
\begin{gather*}%\label{tq2per}
g_{n-i,n+i}:=(-1)^i g_{n,n}.
\end{gather*}
From \eqref{tdiff} with $i=2n$ (and the same equation for $f^E_{2n}$) it follows that under the induction assumption $\p_{\zb} (g_{0,2n})=0$. Hence it is a constant. Therefore we can make this constant to be equal zero, i.e., $g_{0,2n}=0$ by proper choice of the coefficient~$e_n$.

If equation $g_{0,2n}=0$ is satisfied then using the equation
\begin{gather*}%\label{tq5per}
\p_{z} g_{n,n}=2g_{n,n+1}
\end{gather*}
we get that $g_{0,2n+1}=0$. The induction step is completed and the lemma is proved.
\end{proof}
The theorem is proved.
\end{proof}

\subsection{The algebraic spectral curve} \label{ss:curve}

Let $\psi$ be the formal Bloch self-dual solution of the Schr\"odinger operator with the potential $u$ defined by~\eqref{u} from a complex solution $q(z,{\zb})$ of the $O(N)$ sigma model on ${\Sigma} = T^{2}$. Our next goal is to show that $\psi$ is a common eigenfunction for a ring ${\CalA}$ of commuting differential operators in one variable, which we can take to be $z$, or $\zb$. Therefore, it is an expansion near a~marked point of the {\it non-formal} Baker--Akhiezer function on the {\it spectral curve} of ${\CalA}$.

Let ${\bf q}(z,\zb,t),\ t=(t_1,t_2,\dots)$ be a an orbit of ${\bf q}(z,\zb,0)$ under the flows of the NV hierarchy.
As we have shown above, the constraint~\eqref{const}, and thus the equation~\eqref{u} are invariant under the flows.
Taking the $t_{n}$ derivative of \eqref{shrod} and using the equation
\begin{gather}\label{udot}
\p_{t_n} u = - \big( {\pa}_z( {\cL}^{2n+1}_{+} {\bf q}), {\pa}_{\zb} {\bf q} \big) - \big( {\pa}_z {\bf q}, {\pa}_{\zb} \big( {\cL}_{+}^{2n+1} {\bf q}\big) \big)
\end{gather}
we get the equation~\eqref{var} with $L_n = {\cL}^{2n+1}_+$. The space of solutions to an elliptic equation on~$\Sigma$ is finite dimensional. Hence for all but a finite number of integers $n$ there are constants~$c_{n,m}$ such that for the operator
\begin{gather}\label{LLL}
{\tilde L}_{n} := {\cL}^{2n+1}_{+} + \sum_{i=0}^{n-1} c_{n,m} {\cL}^{2m+1}_{+},
\end{gather}
the equation
\begin{gather}\label{LLq}
{\tilde L}_n {\bf q} = 0
\end{gather}
holds. We call the finite set $I$ of indices $n$ for which there are no such constants the \emph{gap set}.

\begin{lem}\label{lm:comring}
The operators ${\tilde L}_n$ defined above commute with each other
\begin{gather*}%\label{com}
\big[ {\tilde L}_n , {\tilde L}_m \big] = 0 , \qquad n,m \notin I.
\end{gather*}
\end{lem}

\begin{proof} Consider the linear combination of the vector fields ${\pa}_{t_{m}}$
\begin{gather*}%\label{vect}
\p_{ {\tilde t}_n} = {\pa}_{t_{n}} + \sum_{m=0}^{n-1} c_{n,m} {\pa}_{t_{m}}
\end{gather*}
with $c_{n,m}$ as in equation~\eqref{LLL}. From the equations~(\ref{udot}), \eqref{LLq} it follows that $\p_{ \tilde t_n} u=0$. Then~\eqref{manu} implies
$\p_{\zb} {\widetilde F}_n=0$, where
\begin{gather*}%\label{Fhat}
\widetilde F_n:=F_{2n+1}^{(1)}+\sum_{m=0}^{n-1} c_{n,m} F_{2m+1}^{(1)}
\end{gather*}
with $F_{2n+1}^{(1)}$ given by the equation~\eqref{res1}. Therefore $\widetilde F_n$ is a constant. From the equation~\eqref{JQ1} it follows that $\widetilde F_n$ is $z$-derivative of some double periodic function. The latter implies $\widetilde F_n=0$. Then from equation~\eqref{sb3} it follows that
\begin{gather}\label{Leigen}
H(L_n\psi)=0.
\end{gather}
The coefficients of the operator $L_n$ are double periodic functions. Therefore, from~\eqref{Leigen} it follows that $k^{-2n+1}L_n\psi$
is a formal Bloch solution of the Schr\"odinger equation. From~\eqref{adjL} it follows that it is a self-dual Bloch solution. Then the uniqueness up to multiplication by a~constant series of the latter implies the equation
\begin{gather}\label{eigenagain}
\tilde L_n \psi = a_{n}(k)\psi,\qquad a_n(k)=k^{2n+1}+\sum_{s=-n}^\infty a_{n,s} k^{-2s-1},
\end{gather}
where $a_{n,s} \in {\BC}$ are some constants.

The equation~\eqref{eigenagain} implies that $\big[\tilde L_n,\tilde L_m\big]\psi=0$. From the form of $\psi$ it easily follows that~$\psi$ is not in a kernel of any nonzero ordinary differential operator in the variable~$z$. The lemma is proved.
\end{proof}

Recall the fundamental fact of the theory of commuting linear
ordinary differential operators \cite{ch1,ch2,kr1,kr2,mum}:
\begin{lem}[Burchnall--Chaundy]\label{lem:BCh}
Let $L_n$ and $L_m$ be commuting ordinary linear differential operators of orders $n$ and $m$, respectively. Then there exists a polynomial $R$ in two variables such that the equation
\begin{gather}\label{BC}
R(L_n,L_m)=0
\end{gather}
holds.
\end{lem}
Recall that the affine curve defined by \eqref{BC} is compactified by one \emph{smooth} point $P_{+}$. The corresponding algebraic curve $\G$ is called the {\it spectral} curve. The maximal commutative ring~${\A}_z$ of ordinary differential operators in the $z$-variable containing the operators $L_n$ and $L_m$ is isomorphic to the ring $A(\G,P_+)$ of meromorphic functions
on $\G$ having the only pole at $P_{+}$.

As shown in \cite{kr1,kr2} for a generic pair of commuting operators of {\it co-prime} orders the spectral curve is smooth and the common eigenfunction $\psi$ of the commuting operators is the Baker--Akhiezer function:
\begin{enumerate}\itemsep=0pt
\item[$1^0$.] {\it As a function of $p\in \G$ it is meromorphic on $\G\setminus P_+$ with $z$-independent divisor $D$ of poles of degree equal to the genus of $\G$.}
\item[$2^0$.] {\it In the neighborhood of $P_+$ it has the form \eqref{psi}.}
\end{enumerate}

\emph{A priori}, there are two spectral curves in the problem under consideration: the one, which we have just discussed, and the other, with a smooth marked point $P_-$, corresponding to the commuting differential operators in the variable~${\zb}$. In fact, these spectral curves coincide. Indeed, the operators $\bar L_n$ satisfy the equations
$\big[\bar L_n,H\big]=\bar B_n H$ where $\bar B_n$ is a differential operator in the variable $\zb$. Then from the uniqueness of the Bloch BA solution up to a~transformation of the form~\eqref{psitrans} the equation
\begin{gather}\label{barLpsi}
\bar L_n\psi=\bar a_n (k) \psi, \qquad \bar
a_n(k)=\sum_{s=0}^\infty \bar a_{n,s} k^{-2s-1}
\end{gather}
follows.
The series $\bar a_n(k)$ in \eqref{barLpsi} is an expansion in the neighborhood of $P_+$ of the function $\bar a_n\in A(\G,P_-)$ having pole of order~$n$ at the second marked point~$P_-$.

\begin{dfn} A Schr\"odinger operator is called algebraic-geometric (finite-gap) if it is stationary for all but a finite number of the NV hierarchy flows.
\end{dfn}

\begin{theo} The spectral curve $\G$ corresponding to an algebraically integrable operator $H$ is a~curve with involution $\s$ fixing the smooth marked points $P_\pm$. In the generic case when $\G$ is smooth: the points $P_{\pm}$ are the only fixed points of $\s$; the common eigenfunction $\psi(z,{\zb} ;p), p\in \G,$ of the operators $L_n$, $\bar L_m$ satisfying the equation $H\psi=0$ is the {\it two-point} BA function on an algebraic spectral curve $\G$, i.e.,
\begin{enumerate}\itemsep=0pt
\item[$1^0$.] In the neighborhoods of $P_{\pm}$ it has an essential singularity of the form
\begin{gather}
\begin{split}
& {\psi}_{+} ( z,{\zb}; p ) = {\rm e}^{ k_{+} z } \left(\sum_{n=0}^{\infty}\xi_{n}^{+}(z,\zb)k_{+}^{-n}\right) , \qquad p \to P_{+},\\
& {\psi}_{-} ( z,{\zb}; p ) = {\rm e}^{ k_{-} {\zb} } \left(\sum_{n=0}^{\infty}\xi_{n}^{-}(z,\zb)k_{-}^{-n}\right) , \qquad p \to P_{-}
\end{split}\label{eq:psik1}
\end{gather}
with
\begin{gather*}%\label{xinorm}
\xi_0^\pm=1.
\end{gather*}
\item[$2^0$.] Outside marked points $P_\pm$ it is meromorphic with $(z,\zb)$-independent divisor of poles $D$ satisfying the constraint
\begin{gather}\label{dadmis}
D+D^\s=\K+P_++P_-,
\end{gather}
where $\K$ is the canonical class, i.e., the equivalence class of the zero divisor of a holomorphic differential on~$\G$.
\end{enumerate}
\end{theo}

The involution $\s$ of $\G$ is just an avatar of the self-duality of the formal Bloch solution of the Schr\"odinger equation or equivalently of the equation~\eqref{adjL}. The statement that in a generic case $\G$ is smooth and $P_\pm$ are the only fixed points of~$\s$ is corollary of the NV reconstruction of the algebraically integrable potentials \cite{nv1,nv2} from a smooth algebraic curve with involution having two fixed point and a divisor~$D$ satisfying~\eqref{dadmis}.

The proof of the constraint \eqref{dadmis} follows the line of arguments
in the proof of Lemma~2.3 in \cite{kr87} (see more details in \cite{KN}).
 The $\Gamma$-derivative $d{\psi}$ is
a solution of the same Schr\"odinger equation \eqref{shrod}. Define the new current
\[
{\beta} = {\psi}^{\s} {\rm d}_{\Sigma} ( {\rm d}{\psi} ) - \big({\rm d}_{\Sigma} {\psi}^{\s} \big) {\rm d}{\psi} ,
\]
a $2$-form on ${\Sigma} \times \Gamma$, with one leg along~$\Sigma$, one leg along~$\Gamma$. From~\eqref{shrod} we derive
\begin{gather}
{\rm d}_{\Sigma} \star \beta = 0 .\label{eq:betacons}
\end{gather}
Since $\psi$ is a Bloch function, its $\Gamma$-differential has the following monodromy properties
(here $q \in \Gamma$ is a point on the Fermi-curve, not to be confused with the vector ${\bf q} \in {\BC}^{N}$):
\begin{gather*}%\label{mondpsi}
{\rm d}{\psi} ( z+ {\omega}_{\a}, {\zb}+{\bar\omega}_{\a}, q ) = w_{\a}(q) ( {\rm d}\psi(z,\zb,q)+{\rm d}p_{\a}(q) \psi(z,\zb,q) ) ,
\end{gather*}
where ${\a} = x, y$,
\begin{gather*}%\label{eq:dlogw}
{\rm d} p_{\a}:={\rm d} \log w_{\a} ,
\end{gather*}
or, in terms of the currents $\beta$ and $j$:
\begin{gather}
T_{\alpha}^{*} ( \star {\beta} ) - \star {\beta} = {\rm d}p_{\alpha} \wedge \star j ,
\label{eq:deckbeta}
\end{gather}
where $T_{\alpha}$, $\alpha = x, y$, are the deck transformations~\eqref{eq:deck}.
Now, integrate \eqref{eq:betacons} over a fundamental domain in ${\tilde\Sigma}$ for ${\BZ}^{2}$, use Stokes theorem to write it as an integral over the boundary, and use~\eqref{eq:deckbeta}, to write
\begin{gather}\label{eq:abcycles}
{\rm d} p_{y} \oint_{A} \star j - {\rm d} p_{x} \oint_{B} \star j = 0.
\end{gather}
In components \eqref{eq:abcycles} reads
\[
{\rm d}p_{y} \langle j_{z} - j_{\zb} \rangle_{x} = {\rm d}p_{y} \langle {\tau} j_{z} - {\bar\tau} j_{\zb} \rangle_{y} ,
\]
or, equivalently
\begin{gather}\label{Differential0}
 ( {\tau}_{1} \langle j_{x} \rangle_{x} - \langle j_{y} \rangle_{x} ) {\rm d}p_{y} =
 ( {\tau}{\bar\tau} \langle j_{x} \rangle_{y} - {\tau}_{1} \langle j_{y} \rangle_{y} ) {\rm d}p_{x},
\end{gather}
where we denote by $\langle \dots \rangle_{\a}$ the averages
\[
\langle {\CalO} \rangle_{x} (y) = \int_{0}^{1} {\rm d}x {\CalO}(x,y) , \qquad \langle {\CalO} \rangle_{y} (x) = \int_{0}^{1} {\rm d}y {\CalO}(x,y).
\]
A priori, in the equation~\eqref{Differential0} the averages are taken at the specific values of $y$ and $x$, respectively, corresponding to the boundaries of the chosen fundamental region. However, the current conservation~\eqref{eq:jconservation} implies the contour integral $\oint_{\gamma} \star j$ does not change
under the homotopy of $\gamma$. Thus, the specific linear combinations of $\langle j_{\beta} \rangle_{\alpha}$ with ${\alpha}, {\beta} = x, y$ we see in~\eqref{Differential0} are independent of~$x$,~$y$.

The equation~\eqref{Differential0} implies that the zeroes of the meromorphic function $\langle {\tau}{\bar\tau} j_{x} - {\tau}_{1} j_{y} \rangle_{y}$ are the zeros of the differential ${\rm d}p_y$, while the zeros of
the meromorphic function ${\langle {\tau}_{1} j_{x} - j_{y} \rangle_x}$ are the zeros of the differential ${\rm d}p_x$. Hence, the differential
\begin{gather*}%\label{Differential}
{\rm d}\Omega:=\frac{-2{\ii} \tau_2 {\rm d}p_x}{\langle {\tau}_{1} j_{x} - j_{y} \rangle_x}=\frac{-2{\ii}\tau_2 {\rm d}p_y}{\langle {\tau}{\bar\tau} j_{x} - {\tau}_{1} j_{y} \rangle_y}
\end{gather*}
is holomorphic on $\G$ away of the marked points $P_\pm$, where it has simple poles with residues~$\mp 1$. Its zeros are the poles of the functions $\psi$ and $\psi^\s$. That completes the proof of~\eqref{dadmis}.
\begin{rem} We call the divisors $D$ satisfying~\eqref{dadmis} {\it admissible}. They are parameterized by points of the Prym $\P(\G)$ variety which is the subvariety of the Jacobian~$J(\G)$ that is odd with respect to involution
of~$J(\G)$ induced by the involution~$\s$ of the curve.
\end{rem}

\begin{rem} For singular spectral curves the differential ${\rm d}\Omega$ is a meromorphic section of the dualizing sheaf with simple poles at $P_\pm$, i.e.,
\begin{gather} {\rm d}\Omega\in H^{0} ( {\G} , \omega_{\G}+P_{+}+ P_{-} ). \label{eq:regdo}
\end{gather}
In more practical terms~\eqref{eq:regdo} means that the preimage of ${\rm d}\Omega$ on the normalization of a singular curve may have simple poles at the preimages of nodes with opposite residues.
\end{rem}

\subsection{The periodicity constraint} \label{ss:period}

The reconstruction of an algebraically integrable potential $u(z,\zb)$ from a smooth algebraic curve~$\G$ with involution having two fixed points~$P_\pm$ and an admissible divisor~$D$ gives in general~$u$ that is a quasi-periodic function
of its arguments. The curves corresponding to double periodic functions are
singled out as follows.

It follows from the non-degeneracy of the imaginary part of the matrix of $b$-periods of normalized holomorphic differentials that for a curve ${\G}$ with a fixed first jet $\big[k_\pm^{-1}\big]_1$ of local coordinates at the two marked points $P_{\pm}$ there exist unique meromorphic differentials ${\rm d}p_{\a}$, with ${\a}=x,y$ having a pole of order~$2$ at~$P_{+}$ of the form
\begin{gather*}%\label{imp}
{\rm d}p_{x} = {\rm d}k_{+} \big( 1+O\big(k_{+}^{-2}\big) \big) , \qquad {\rm d}p_{y} = {\tau} {\rm d}k_{+} \big( 1+O\big(k_{+}^{-2}\big) \big) ,
\end{gather*}
and a pole of order $2$ at $P_{-}$ of the form
\begin{gather*}%\label{imp1}
{\rm d}p_{x} = {\rm d}k_{-} \big( 1+O\big(k_{-}^{-2}\big)\big) , \qquad {\rm d}p_{y} = {\bar\tau} {\rm d}k_{-} \big( 1+O\big(k_{-}^{-2}\big) \big) ,
\end{gather*}
holomorphic away from $P_\pm$, and such that their periods are \emph{imaginary}, i.e., $\operatorname{Re} \big( {\oint}_{c} {\rm d}p_\a\big)=0$ for any $1$-cycle $c$ on~$\G$.
The locus of spectral curves corresponding to the double periodic potentials is singled out by the constraint that the periods are integer multiple of $2\pi \ii$,
\begin{gather}\label{imp2}
n_{\a,c}:= \frac{1}{2\pi\ii} \oint_c {\rm d}p_{\a} \in \mathbb Z,\qquad \forall c\in H_1(\G, \mathbb Z).
\end{gather}
Indeed, if equations \eqref{imp2} hold then the functions $w_\a(p)=\int^p {\rm d}p_\a$ are single valued on $\G$. Then the uniqueness of the Baker--Akhiezer function implies
\begin{gather*}%\label{bloch25}
{\psi} ( z+\omega_{\a},\zb + \bar \omega_{\a}; p t) = w_{\a} (p) {\psi} ( z,\zb; p ) , \qquad p\in \G,
\end{gather*}
since both sides of the equation have the same analytical properties on~$\G$.

Choosing a basis of cycles on $\G$ one can extend it to a~basis of cycle on the curves in a~neighborhood of~$G$. Then for fixed set of $n_{\a,c}$ equations~\eqref{imp2} identify the spectral curves of periodic potentials with the level set of~$4g$ real analytic functions on the moduli space of curves with fixed first jets of local coordinates near two marked points.

A smooth genus~$g$ algebraic curve with involution having $2n+2$ fixed points is uniquely defined by a factor-curve $\G_0=\G/\s$ and a choice of $2n+2$ points of it. Hence the space of such curves is of dimension $3g_0+2n-1$, where $g_0$ is the genus of $\G_0$. The Riemann--Hurwitz formula implies $g=2g_0+n$. If the first jets of local coordinates near two fixed points of involution are odd, i.e., $\s^* {\rm d}p_\a=-{\rm d}p_\a$. Hence, their periods over cycles that are even with respect~$\s$ vanishes. The latter implies that the locus of the algebraic data $\big\{\G,\s, P_\pm, \big[k_\pm^{-1}\big]_1\big\}$, with $\s$ having $2(n+1)$ fixed points, corresponding to periodic potentials is of dimension $g_0+1$.

Let $\mathcal S^{g_0,n}$ of curves $\G$ be the locus of algebraic geometrical data satisfying periodicity constraints for some fixed set $n=\{n_{\a,c}\}$. If we choose a basis of $a$ and $b$ cycles on $G_0$ with the canonical matrix of intersection, then the (local) coordinates on ${\mathcal S}^{g_0,n}$ can be defined similarly to those for the families of the Seiberg--Witten curves, namely
\begin{gather}\label{Aper}
A_{0} = \operatorname{res}_{P_+} p_x {\rm d}p_y, \qquad A_{i} = \oint_{\pi^*(a_i)} p_x {\rm d}p_y , \qquad i=1,\dots, g_0,
\end{gather}
where $\pi\colon {\G} \to {\G}_{0}$ is the projection.
Although the Abelian integral $p_x$ of ${\rm d}p_{x}$ is multi-valued the expressions \eqref{Aper} are well-defined. Indeed, a shift of $p_x$ by a constant does not change $A_0$ since ${\rm d}p_y$ has no residue. It does not change $A_i$ either, since ${\rm d}p_y$ is odd with respect to the involution~$\sigma$ while the cycle~$\pi^*(a_i)$ is even.

Note that in \cite{KN} it was shown that
\begin{gather*}%\label{uaver}
A_0=\int_\Sigma u(z,\zb) {\rm d}z\wedge {\rm d}\zb.
\end{gather*}

We conclude this section by identifying algebraic spectral curve with the Fermi curve. More precisely,
\begin{theo}[{\cite[Theorem 3.1]{kr-spec}}]\label{floque} Let $\G$ be a smooth spectral curve corresponding to a double periodic potential~$u$ then it is a normalization of the Bloch--Floquet locus $C_u$ defined in~\eqref{bloch}, i.e.,
there exits a map $\G\to C_u$ which is one-to-one outside the preimage of the set of singular points of~$C_u$.
\end{theo}

\section[w\_infty-harmonic maps to spheres]{$\boldsymbol{w_{\infty}}$-harmonic maps to spheres}\label{s:Winfty}

In this section we generalize the construction of Novikov and Veselov to the case of reducible spectral curves, and show that it provides solutions for the $O(2n+1)$-model. Moreover the periodicity constraint for the potential is effectively solved in terms of the spectral curves of the elliptic Calogero--Moser system.

\begin{rem} We present here only the basic case when each of the irreducible components~$\G_\pm$ is smooth. The generalization of the construction to the case of singular curves follows the standard line of arguments and will be presented elsewhere.
\end{rem}

Let $\Gamma_\pm$ be a smooth genus $g_\pm$ algebraic curve with the holomorphic involution
\begin{gather*}
{\sigma}\colon \ \Gamma_{\pm} \longmapsto \Gamma_{\pm} , %\label{eq:inv1}
\end{gather*}
with $2(n+1)$ fixed points
\begin{gather*}
\sigma (P_{\pm})=P_{\pm} , \qquad {\sigma}\big(p_{\pm}^{i}\big) = p_{\pm}^{i}.%\label{eq:fpinv1}
\end{gather*}
Let us fix the ${\sigma}$-odd local parameters $k_{\pm}^{-1}$ in the neighborhoods of the marked points $P_\pm$,
\begin{gather*}
k_{\pm}(\sigma (p))=-k_{\pm}(p). %\label{eq:oddpar1}
\end{gather*}
The projection
\begin{gather*}
\pi \colon \ \Gamma_\pm \longmapsto {\Gamma}_{\pm}^0 = {\Gamma}_{\pm} / {\sigma} %\label{eq:fcurve1}
\end{gather*}
represents $\Gamma_\pm$ as a two-sheet covering of the quotient-curve $\Gamma_\pm^0$ with $2 (n+1)$ branch points~$P_{\pm}$,~$p_\pm^i$, the involution $\sigma$ permuting the sheets. From the Riemann--Hurwitz formula it follows that
the genus of $\G_\pm$ equals
\begin{gather*}
{g_\pm}=2g^0_\pm+n, %\label{eq:genera1}
\end{gather*}
where $g^0_\pm$ is the genus of $\Gamma^0_\pm$.

Let ${\rm d}{\Omega}_{\pm} (p)$ be a third kind meromorphic differential on $\Gamma^0_\pm$ with the divisor of poles at the branching locus ${\Gamma}_{\pm}^{\sigma}$ with residues $\mp 1$ at the marked points~$P_\pm$.
The differential ${\rm d}{\Omega}_{\pm}$ has
\[ {\#} {\Gamma}_{\pm}^{\sigma} + 2g^0_{\pm}-2 = g_{\pm} + n \]
zeros that we denote
by ${\gamma}_{s}^{0,\pm}$, $s=1, \dots, g_\pm+n$,
\begin{gather}
{\rm d}\Omega_\pm\big(\gamma_s^{0,\pm}\big)=0. \label{eq:diffdo1}
\end{gather}
For each zero ${\gamma}_{s}^{0, \pm}$ we choose one of its preimages on $\Gamma_{\pm}$, i.e., a point $\gamma_s^\pm$ on $\Gamma_\pm$ such that
\begin{gather*}
\pi \big( {\gamma}_s^{\pm} \big)=\gamma_s^{0,\pm}, \qquad s=1,\dots, { g}_{\pm} +n %\label{eq:preim1}
\end{gather*}
(there are $2^{g_\pm+n}$ such choices). Below $D_{\pm} = {\gamma}_1^{\pm}+ \dots +\gamma_{{g}_{\pm}+n}^\pm$ will be called the admissible divisor.

\begin{lem}\label{lm:main} The generic data consisting of a matrix $G\in O(2n+1,{\BC})$,
\begin{gather*}%\label{G}
G^{t}G = 1
\end{gather*}
and a pair $\big( {\Gamma_\pm}, {\sigma}, P_{\pm}, k_{\pm}, p_\pm^i, \gamma_s^{\pm} \big)$,
defines the unique pair of functions $\psi_{\pm}(z,\zb; p)$ on $\Gamma_\pm$ with the following analytic properties in the variables $p\in {\G}_{\pm}$, respectively:
\begin{enumerate}\itemsep=0pt
\item[$1^0$.] Outside $P_{\pm}$ the only singularities of the function $\psi_\pm$ are the poles at $\gamma_s^\pm$, $s=1,\dots, {g}_\pm+n$. The poles are simple if all of the ${\gamma}^{\pm}_{s}$'s are distinct.
\item[$2^0$.] In the neighborhoods of $P_{\pm}$ the function $\psi_\pm$ has an essential singularity of the form~\eqref{eq:psik1}
with
\begin{gather*}%\label{xinorm-bis}
\xi_0^\pm=1.
\end{gather*}

\item[$3^0$.] The gluing equations
\begin{gather}\label{gluing}
{\bf y}_{+} (z,\zb) = G {\bf y}_{-} (z,\zb),
\end{gather}
where ${\bf y}_\pm = \big( y_{\pm}^{i} \big)_{i=1}^{2n+1} \in {\BC}^{2n+1}$ are the vectors with the coordinates
\begin{gather}\label{y}
y_\pm^{i} = r_{\pm}^i\psi_{\pm} \big( z,{\zb}; p_{\pm}^{i} \big) , \qquad i=1,\dots, 2n+1,
\end{gather}
with
\begin{gather}\label{Rdef}
\big(r_{\pm}^i\big)^2=\mp \operatorname{res}_{p_{\pm^i}} {\rm d}\Omega_\pm
\end{gather}
hold.
\end{enumerate}
\end{lem}
\begin{proof} The first step of the proof is a simple counting of parameters and equations. According to~\cite{kr1} the space of functions $\psi_{\pm}$ satisfying \eqref{eq:psik1} is $n+1$-dimensional. Therefore, the $(2n+1)$ gluing equations \eqref{gluing} cut out a one-dimensional space of pairs ${\psi} = ( {\psi}_{+}, {\psi}_{-} )$.
The condition $\xi_{0}^{+}=1$ further normalizes $\psi$. It remains to show that the second normalization $\xi_{0}^{-}=1$ is satisfied automatically.

Consider the differentials ${\rm d}{\widetilde \Omega}_{\pm} = {\psi}_{\pm}^{\sigma} \psi_{\pm} {\rm d}{\Omega}_{\pm}$ on ${\Gamma}_{\pm}$, respectively. By definition of the admissible divisor ${\rm d}{\widetilde \Omega}_{\pm}$ is a meromorphic differential on $\Gamma_{\pm}$ with only the simple poles at~$P_\pm$,~$p_\pm^i$. Then by the residue theorem we have
\begin{gather}
\big(\xi_0^{+}\big)^2 = \sum_{i=1}^{2n+1} {\res}_{p^i_+}{\rm d}{\widetilde \Omega_{+}} = -{\bf y}_{+}^{t} {\bf y}_{+}
 = -{\bf y}_{-}^{t} {\bf y}_{-} = {\res}_{P_{-}}{\rm d}{\widetilde {\Omega}_{-}} = \big({\xi}_{0}^{-}\big)^2, \label{ressum1}
\end{gather}
where in the middle we used the equations~(\ref{y}), (\ref{Rdef}) and the assumption $G\in O(2n+1,{\BC})$.
Since at $z=0$, $\zb=0$ we have $\xi_0^{\pm}(0,0)=1$ the equation~\eqref{ressum1} implies the equality $\xi_{0}^{+}=\xi_{0}^{-}$. The lemma is proved.
\end{proof}

In what follows we call the pair of functions ${\psi}:= ( {\psi}_{+} , {\psi}_{-} )$ the BA function on $\Gamma:= \Gamma_{+}\bigsqcup\Gamma_{-}$. Given $\psi$ and a matrix $H\in O(2n+1,{\BC})$ introduce the vector ${\bf q}(z,\zb) \in {\BC}^{2n+1}$:
\begin{gather}\label{Xmay}
{\bf q}=H {\bf y}_{+}.
\end{gather}
\begin{theo}\label{thm:main}
The BA function $\psi(z,\zb; p)$ on~$\Gamma$ satisfies the equation
\begin{gather*}
 ( {\pa}_z{\p}_{\zb} - u(z,\zb) )\psi(z,{\zb}; p)=0, %\label{eq:schr1}
\end{gather*}
with the potential
\begin{gather*}
u(z,{\zb}) = {\pa}_{\zb}{\xi}_{1}^{+} = {\pa}_z {\xi}_{1}^{-}. %\label{eq:potu1}
\end{gather*}
Moreover, the $2n+1$-dimensional vector ${\bf q}(z,\zb)$ defined in~\eqref{Xmay} satisfies the equations:
\begin{gather}\label{sphere}
({\bf q}, {\bf q}) =1,\\
\label{conformal}
\big(\pa^{i}_z {\bf q}, {\bf q}\big)=0 ,\qquad i > 0.
\end{gather}
\end{theo}

\begin{proof} The proof of the first statement of the theorem is standard and based on the uniqueness of the BA function. Equation~(\ref{sphere}) is just the first equation in~(\ref{ressum1}).

The differential $\big({\pa}^i_{z} {\psi}_{-}\big) {\psi}_{-}^{\sigma}{\rm d}\Omega_{-}$ has no residue at $P_-$ for $i>0$. Hence the sum of its residues at the points $p_{-}^i$, $i=1,\dots, 2n+1$ is zero. The latter is equivalent to the equation~\eqref{conformal} due to the gluing conditions~\eqref{gluing} and orthogonality of~$H$.
\end{proof}

\subsection{Real and regular solutions}

In general the potential $u(z, {\zb})$, and the functions $q^{i}(z, {\zb})$, given by the construction above, may have poles. The following theorem describes sufficient conditions for $q^{i}(z,\zb)$ to be real and regular.

\begin{theo}\label{real_case} Suppose the data in Lemma~{\rm \ref{lm:main}} satisfies the following conditions:
\begin{enumerate}\itemsep=0pt
\item[$(a)$] there is an anti-holomorphic bijection ${\tau}\colon {\Gamma}_{+}\longrightarrow {\Gamma}_{-}$ commuting with $\sigma$, such that
\begin{gather*}
{\tau} (P_{+}) = P_{-} , \qquad {\tau}\big( p_{+}^{i} \big) = p_{-}^{i},%\label{eq:tauinv}
\end{gather*}
\item[$(b)$] the matrix $G$, which is be definition orthogonal, is Hermitian, i.e.,
\begin{gather}\label{newconst}
{\bar G} = G^{-1} = G^t,
\end{gather}
\item[$(c)$] the admissible divisors are $\tau$-invariant, e.g., ${\tau} ( D_{+} ) = D_{-}$.
\end{enumerate}
Then the potential defined by the corresponding BA function is real, $u=\bar u$.

Moreover, if
\begin{gather}\label{GH}
G = H^{-1} {\bar H}
\end{gather}
for some $H\in O(2n+1,{\BC})$, then the vector ${\bf q}(z,\zb)$ given by~\eqref{Xmay} is real and therefore regular.
\end{theo}

\begin{rem}In this case the curve ${\Gamma}_{-}$ can be seen as ${\Gamma}_{+}$
with the opposite complex structure, so that $\Gamma=\Gamma \cup \bar \Gamma$ has a real structure.
\end{rem}

\begin{proof} For a generic effective divisor $D_\pm$ of degree $g_\pm+n$ equation \eqref{eq:diffdo1} uniquely defines the corresponding meromorphic differential ${\rm d}\Omega_\pm$. Therefore, under the assumption $(c)$ the equation
\begin{gather*}%\label{omegabar}
{\rm d}\Omega_+(p)=-\overline {{\rm d}\Omega}_-(\tau(p))
\end{gather*}
holds. The latter implies the equation $\big(r^i_\pm\big)^2=\overline{\big(r^i_\mp\big)}^2$ for $\big(r_\pm^i\big)^2$ defined in~\eqref{Rdef}. It is assumed that the square roots of the both sides of the equation are chosen consistently, i.e., that the equations
\begin{gather}\label{sqres}
r^i_\pm=\bar r^i_\mp
\end{gather}
hold.

The equations~\eqref{newconst}, \eqref{sqres} imply that the gluing equations \eqref{gluing} hold for the pair
\[
\big( {\bar\psi}_{-} ( z,{\zb} ; \tau(p) ), {\bar\psi}_{+}( z,{\zb} ; {\tau}(p))\big)
\]
 of BA functions. The uniqueness of BA function then implies
\begin{gather}\label{reality1}
{\psi}_{\pm} ( z,{\zb}; p ) = {\bar\psi}_{\mp} ( z,\zb; \tau(p) ) .
\end{gather}
The reality of $u$ is a direct corollary of \eqref{reality1}. The reality of $\bf q$ is an easy corollary of~\eqref{reality1} and~\eqref{GH}. The theorem is proved.
\end{proof}

\subsection{Periodicity constraint} The potential $u(z,\zb)$ defined by data in Lemma~\ref{lm:main} is periodic if and only if there are functions~$w_{\a,\pm}$ on~$\Gamma_\pm$ such that equations
\begin{gather}\label{blochpm}
{\psi}_{\pm} ( z+\omega_\a,\zb+\bar\omega_{\a} ; p ) = w_{\a}^{\pm}(p) {\psi}_{\pm} ( z,{\zb}; p )
\end{gather}
hold. The first part of the constraints that single out the corresponding algebraic-geometrical data are similar to that in Section~\ref{ss:period}.

Namely, denote by ${\rm d}p_\a^+$ the unique differentials on $\Gamma_{+}$ with a single pole at $P_+$ of the form
\begin{gather*}%\label{impa}
{\rm d}p^{+}_{x} = {\rm d}k_{+} \big( 1 + O\big(k_{+}^{-2}\big) \big) ,\qquad {\rm d}p^{+}_{y} = \tau {\rm d}k_{+} \big( 1 + O\big(k_{+}^{-2}\big)\big)
\end{gather*}
and denote by ${\rm d}p_\a^{-}$ the unique differentials on $\Gamma_{-}$ with a single pole at $P_{-}$ of the form
\begin{gather*}%\label{imp1a}
{\rm d}p_{x}^{-} = {\rm d}k_{-} \big( 1+O\big(k_{-}^{-2}\big)\big) , \qquad {\rm d}p_{y}^{-} = {\bar\tau} {\rm d}k_{-} \big( 1+O\big(k_{-}^{-2}\big)\big) ,
\end{gather*}
such that all their periods are {\it imaginary}.

The functions
\begin{gather}\label{abpm}
w_\a^{\pm}(p)=\exp\left(\int^p {\rm d}p_\a^\pm\right)
\end{gather}
are single-valued on $\Gamma_\pm$ iff the periods of $dp_{\a}^{\pm}$ are integral multiples of $2\pi\ii$
\begin{gather}\label{periodicity1}
\frac{1}{2\pi\ii} \oint_{c^\pm} {\rm d}p_\a^{\pm} = n_{\a,c^\pm}^{\pm}\in {\BZ} \qquad\forall {c^\pm} \in H_{1} ( {\Gamma_\pm}, {\BZ} ).
\end{gather}
At first glance the equations~\eqref{periodicity1} are identical to the equations~\eqref{imp2}.

\subsubsection{Enters elliptic Calogero--Moser system}
It turns out that in the case of reducible Fermi curves the constraints \eqref{periodicity1} are solved by the spectral curves of elliptic Calogero--Moser (eCM) system!

\begin{theo}[\cite{kr-grush}] The equations~\eqref{periodicity1} are satisfied iff ${\Gamma}_{\pm}$ is the normalization of the spectral curve ${\C}_{N_{\pm}}$ of $N_{\pm}$-particle eCM system.
\end{theo}
Recall that the $N$-particle eCM is the Hamiltonian system on
\begin{gather*}
{\X}_{N} = T^* \big(E^{N} \backslash {\rm diag}\big) \big\{ ({\rho}_{i}, z_{i})_{i=1}^{N} \,|\, {\rho}_{i} \in {\BC}, \,z_{i} \in E ,\, z_{i} \neq z_{j} ,\, i \neq j \big\},
\end{gather*}
$E = {\BC}/{\BZ}{\omega}_{x} \oplus {\BZ}{\omega}_{y}$, which is governed by the Hamiltonian
\begin{gather*}
H_{2} = \frac 12 \sum_{i=1}^{N} {\rho}_{i}^{2} + {\nu}^{2} \sum_{i < j} {\wp}(z_{i} - z_{j}).
%\label{eq:ecmham}
\end{gather*}
This is an algebraic integrable system with the Lax operator~\cite{kr-cm}
\begin{gather}\label{CMLax}
L({\alpha}) = \left\Vert {\rho}_{i} {\delta}_{ij} + {\nu} \frac{{\sigma} ({\alpha} + z_{i}- z_{j})}{{\sigma} (z_{i}- z_{j}){\sigma} ({\alpha})} (1 - {\delta}_{ij}) \right\Vert_{i,j =1}^{N}.
\end{gather}
The flows of eCM linearize on the Jacobian of the spectral curve
\begin{gather}
C_{N} \subset M = \overline{{\BC} \times ( E \backslash \{ {\alpha} = 0 \} )} , \nonumber\\
C_{N} \colon \ \operatorname{Det} ( k - L({\alpha}) ) = k^N + \sum_{j=1}^{N} k^{N-j} c_{j}({\alpha}) = 0 .\label{characteristic}
\end{gather}
The coefficients $c_{j}({\alpha})$, $j = 1, \dots , N$ of the characteristic polynomial~\eqref{characteristic} are meromorphic functions on $E$ with poles at $\alpha = 0$ of order~$j$, so that $c_{1}({\alpha}) = H_{1}$, $c_{2}({\alpha}) = H_{2} - {\nu}^{2} \frac{N(N-1)}{2} {\wp}({\alpha})$, etc.

In what follows the value of $\nu$ is immaterial, so we set it to be equal to $\nu = 1$.
As shown in~\cite{kr-cm}, near $\alpha=0$ the polynomial $R(k,\alpha)$
admits a factorization of the form
\begin{gather*}%\label{r1}
R(k,{\alpha}) = \prod_{i=1}^{N}\big( k+a_i\alpha^{-1}+h_i+O(\alpha) \big),
\end{gather*}
with $a_1=1-N$ and $a_i=1$ for $i>1$, for some $h_i\in {\BC}$. This implies that the closure of the affine curve defined by~\eqref{characteristic} is obtained by adding one point $( {\infty},0)$, at which $N-1$ branches are tangent to each other (corresponding to $a_{2}=\dots=a_{N}=1$), and one branch is transverse to them. Thus if we blow up the point $( {\infty}, 0)\in {\BP}^{1} \times E$, we would get a smooth point $P$ corresponding to the first branch, and a point~$p'$ contained in the remaining $N-1$ branches. Thus, generically, after the second blow up we get a compact curve whose projection onto the elliptic curve has~$N$ preimages. We call the marked point the preimage of~$P$ on the second blowup.

Due to the degenerate nature of the residue of $L({\alpha})$ at ${\alpha}=0$ there are only $N$ independent linear parameters in the coefficient functions. An explicit form of these parameters proposed in~\cite{pd} is based on an observation that a polynomial $R(k,\alpha)$ has a unique representation of the form
\begin{gather}\label{phdoker1}
 R(k,\alpha):=f(k+\zeta(\alpha),\alpha),
\end{gather}
where
\begin{gather}\label{ph}
f(p,\alpha)=\frac{1}{\sigma(\alpha)}\ \sigma\left(\alpha+\frac{\partial}{\partial p} \right)H(p)=
\frac{1}{\sigma(\alpha)}\ \sum_{n=0}^N \frac{1}{ n!} \partial_\alpha^{ n} \sigma(\alpha) \frac{\partial^n H}{\partial p^n} ,
\end{gather}
and $H$ is the monic degree $N$ polynomial
\begin{gather}\label{ph1}
 H(p) = p^{N} + \sum_{l=1}^{N} I_{l} p^{N-l},
\end{gather}
whose coefficients ${\bf I} = ( I_{1},\dots,I_{N} )\in {\B} = {\BC}^{N}$ are the integrals of motion of the $N$-particle eCM system.

For generic ${\bf I} \in {\B}$ the corresponding curve $C_{N}$ is smooth of genus $N$. For $\bf I$ belonging to certain loci in ${\B}$ the curves $C_N$ degenerate.
\begin{dfn} We call a smooth genus $g$ algebraic curve \emph{an $N$-particle eCM curve}
if it is a~normalization of a~curve defined by the equation~\eqref{characteristic}.
\end{dfn}

By this definition we always have $g\leq N$. As shown in~\cite{kr-grush},
for fixed $g$, the $N$-particle eCM curves become dense in the moduli space of all smooth genus~$g$ algebraic curves, as $N\to \infty$.

The number $N$ of eCM particles can be expressed through
the differentials ${\rm d}p_{x}$, ${\rm d}p_{y}$ on a~smooth genus~$g$ algebraic curve whose periods satisfy the periodicity constraints
(with general~$\nu$)~\eqref{periodicity1}:
\begin{gather*}%\label{Nparticl}
N=\left| \frac{1}{2{\pi} \ii \nu} \operatorname{res}_{P} (p_{x} {\rm d}p_{y} ) \right|.
\end{gather*}
For the eCM curves the multipliers \eqref{abpm} are quite explicit: for $p= (k, {\alpha})$, with ${\zeta} = {\sigma}^{\prime}/{\sigma}$
\begin{gather*}%\label{abformula}
w_x(p) = {\rm e}^{k-{\zeta}(\alpha) + 2 {\zeta}(\frac 12){\alpha}} , \qquad
w_y(p) = {\rm e}^{{\tau} ( k - {\zeta} (\alpha) ) + 2 {\zeta}(\frac {\tau}2){\alpha}}.
\end{gather*}

\subsubsection{Spectral curves of the turning points of the eCM system} Our next goal is to single out the eCM curves with involution. Notice, that since the Weierstrass $\sigma$-function is odd, the eCM curves defined by the equations~(\ref{phdoker1})--(\ref{ph1}) with $I_{\rm odd}=0$ are invariant under the involution
\begin{gather}\label{eq:k3}
{\tilde\sigma}\colon \ (k,{\alpha}) \mapsto (-k, -{\alpha}),
\end{gather}
which leaves the marked point $P$ fixed. The involution $\tilde\sigma$ descends to the involution of $E$ fixing ${\alpha}=0$ and the half-periods ${\alpha} = {\omega}_{j}$, $j = 1, 2,3$, where
\[
{\omega}_{1} = \frac 12 , \qquad {\omega}_{2} = \frac {\tau}2 , \qquad {\omega}_{3} = \frac {\tau +1}2.
\]
Hence $\tilde\sigma$ induces an involution of the fibers on the eCM curve over the half-periods of $E$.

Therefore, for even $N$ there is at least one additional fixed point of $\tilde\sigma$ among the preimages of $\alpha=0$, besides $P$. The same parity argument shows that for odd $N$ there are at least four fixed points of $\tilde\sigma$ - the marked point and one over each of the non-zero half-periods.

The dimension of the space ${\B}_{\rm tp} \subset {\B}$ of the eCM curves invariant under the symmetry~\eqref{eq:k3}
is~$[N/2]$. For generic values of $I_{2i}$, $i=1,\dots,[N/2],$ the corresponding curves are smooth of genus~$N$ on which the fixed points of $\tilde\sigma$ described above are the only fixed points, i.e., generically, $\tilde\sigma$ has $2$ and $4$ fixed points for~$N$ even and odd, respectively.

Now we are going to show that:
\begin{theo} A spectral curve $C_N$ of the eCM system admits a holomorphic involution $\tilde\sigma$ under which the marked point~$P$ is fixed, i.e., ${\tilde\sigma}(P)=P$, iff it corresponds to a \emph{turning point} $(0,z_i)_{i=1}^N \in {\X}_{N}$.
\end{theo}

The \emph{if} part of the statement of the theorem is obvious. Indeed, if ${\rho}_i=0$, for all $i=1,\dots, N,$ then the Lax operator \eqref{CMLax} obeys (cf.\ \cite[equations~(7.53) and Section~7.5]{nikvas}):
\begin{gather*}%\label{oddL}
L(-{\alpha}) = - L^{t}({\alpha})
\end{gather*}
implying that the corresponding spectral curve is invariant under the symmetry~\eqref{eq:k3}.
The \emph{only if} part of the theorem is a corollary of the following new identity for the Riemann theta-function corresponding to a smooth genus $g$ algebraic curve with at least one point fixed, which is an algebraic-geometric incarnation of the characterization of the CKP tau function in terms of turning points of the KP hierarchy (see \cite[Theorem~2.2]{kz}).

\begin{theo} Let $\G$ be a smooth genus~$g$ algebraic curve with involution $\tilde\sigma$ under which the marked point~$P$ is fixed, $P={\tilde\sigma}(P)$. Then for any point $Z_0$ of the Jacobian such that
\begin{gather}\label{Zodd}
Z_0+{\tilde\sigma}(Z_0)=K+2A(P)\in \operatorname{Jac}(\Gamma)
\end{gather}
$(K$ is the canonical class$)$ the identity
\begin{gather}\label{newid}
{\pa}_t {\pa}_z \ln \theta (Uz+Vt+Z_0|B)|_{t=0}=0
\end{gather}
holds.
Here $\theta (Z|B)$ is the Riemann theta-function defined by the matrix $B$ of $b$-periods of a basis of normalized holomorphic differentials on $\G$; the vectors $U$ and $V$ have the coordinates
\begin{gather*}%\label{UV}
U^j=\frac 1{2\pi i}\oint_{b_j} {\rm d}\Omega_2,\qquad V^j=\frac 1{2\pi i}\oint_{b_j} {\rm d}\Omega_3,
\end{gather*}
where ${\rm d}\Omega_2$ and ${\rm d}\Omega_3$ are the normalized meromorphic differentials on $\Gamma$ with the only pole at $P$ of the second and the third order, respectively.
\end{theo}
\begin{proof} Recall that a smooth genus $g$ algebraic curve $\G$ with fixed local coordinate in the neighborhood of a point $P\in \Gamma$, $k^{-1}(P)=0$ and a~generic effective degree~$g$ divisor $D={\gamma}_{1}+ \cdots +{\gamma}_{g}$ defines a unique function $\phi(z,t;p)$ with the following analytic properties with respect to $p\in \Gamma$ (with~$z$,~$t$ considered as fixed parameters):
\begin{enumerate}\itemsep=0pt
\item[$1^0$.] Outside $P$ the singularities of $\phi$ are poles at the divisor $D$.
\item[$2^0$.] In the neighborhood of $P$ the function $\phi$ has the form
\begin{gather*}
\phi(z,t; p)={\rm e}^{k z+k^2 t}\left(1+\sum_{s=1}^{\infty}\xi_{s}(z,t) k^{-s}\right) , \qquad k=k(p). %\label{eq:phikr1a}
\end{gather*}
\end{enumerate}
The explicit formula for $\phi$ in theta functions is \cite{kr1}:
\begin{gather}\label{eq:baf1}
\phi(z,t; p)=\frac{\theta({A}(p)+z U+tV+Z|B) \theta({Z|B})}
{\theta(z U+tV+{Z}|B) \theta({A}(p)+Z|B)} {\rm e}^{z {\Omega}_{2}(p)+t {\Omega}_{3}(p)}.
\end{gather}

\begin{rem} In~\cite{kr1} the basic BA function $\phi( {\bf t};p)$ was defined as the function of an infinite set of variables ${\bf t}=(t_1,t_2,\dots)$. Here we only use its dependence on the first two times $z=t_1$, $t=t_2$, while setting to zero the rest: $t_i=0$, $i>2$. The subspace of $Z$ for which~\eqref{Zodd} holds is one of the connected components of the odd part of the Jacobian. It is invariant under the flows corresponding to the odd times $t_{2j+1}$ of the KP hierarchy.
\end{rem}

The function $\phi$ satisfies the linear equation \cite{kr1}:
\begin{gather}\label{laxkp1}
\big( {\pa}_{t}-{\pa}^{2}_{z} +u (z,t) \big) \phi(z,t; p)=0
\end{gather}
with
\begin{gather*}%\label{umarch}
u=2\pa^2_z \ln \theta(Uz+Vt+Z|B)+2c,
\end{gather*}
where $c$ is a constant equal to the first coefficient of the Laurent expansion ${\rm d}\Omega_2={\rm d}k\big(1+ck^{-2}+\cdots\big)$.

The substitution of the equation~\eqref{eq:baf1} into the equation~\eqref{laxkp1} gives the sequence of equations
\begin{gather}\label{eqxi}
\pa_{t}\xi_{s}-2\pa_{z} \xi_{s+1}-\pa^2_z \xi_s+u\xi_s=0 , \qquad s = 0, 1 , \dots
\end{gather}
of which the first, $s=0$, gives an expression of $u$ in terms of $\xi_1$: $u=2\pa_z\xi_1$.

Suppose now that the curve $\Gamma$ admits an involution $\tilde\sigma$ under which $P$ is fixed and $k$ is odd, $k(p)=-k({\tilde\sigma}(p))$. Suppose also that the divisor $D+{\tilde\sigma}(D)$ is the zero-divisor of a meromorphic differential ${\rm d}\Omega_{*}$ with the only pole at~$P$. The latter is equivalent to~\eqref{Zodd}. Then the differential
\[
{\rm d}\widetilde \Omega_*:=(\pa_z{\psi}(z,0; p)) {\psi}(z,0; {\tilde\sigma}(p)){\rm d}\Omega_*
\]
is a meromorphic differential on $\Gamma$ with the only pole at~$P$. Hence, it has no residue at~$P$. Computing this residue in terms of the coefficients of the expansion~(\ref{eq:baf1}) we get
\begin{gather}\label{resPP}
2\xi_2(z,0)-\xi_1^2(z,0)+\pa_z\xi_1(z,0)+c_1=0,
\end{gather}
where $c_1$ is a constant defined by the Laurent expansion of ${\rm d}\Omega_*$ at $P$.

Taking the $z$-derivative of (\ref{resPP}) and using (\ref{eqxi}) with $s=1$ we get the equation
\begin{gather*}%\label{respp}
0=\pa_{z} \big( 2\xi_2(z,0)-\xi^2_1(z,0)+\pa_z\xi_1(z,0) \big) = {\pa}_{t} {\xi}_{1} |_{t=0},
\end{gather*}
which is equivalent to the equation~\eqref{newid}. The theorem is proved.
\end{proof}

We are now ready to present the full set of constraints on the geometric data above which produce
the solutions of the double-periodic $O(2n+1)$-model.

\begin{theo} \label{prop} The potential $u(z,\zb)$ defined by the data in Lemma~{\rm \ref{lm:main}} is double-periodic if and only if:
\begin{enumerate}\itemsep=0pt
\item[$(i)$] ${\G}_{\pm}$ is a normalization of the spectral curve of \emph{a turning point} of the eCM system;
\item[$(ii)$] the matrix $G$ satisfies the equation
\begin{gather}\label{constG}
G=W_{\a}^{+} G W_{\a}^{-},
\end{gather}
where $W_\a^{\pm}$ are two diagonal matrices
\begin{gather}\label{Wconst}
W_{\a}^{\pm} = \operatorname{diag} \big( w_\a^{\pm}\big(p_{\pm}^{i}\big) \big).
\end{gather}
\end{enumerate}
\end{theo}
The part $(i)$ of the theorem statement was proved above. The part~$(ii)$ easily follow from the compatibility of the monodromy equations~\eqref{blochpm} and the gluing conditions~\eqref{gluing}.

\begin{rem} Note, that under the involution ${\tilde\sigma}$ the functions $w_{\a}^\pm$ transform to $\big(w_{\a}^\pm\big)^{-1}$. Therefore, at the branch points $w_{\a}^{\pm}$ equal $1$ or~$-1$. Thus, if the condition~$(i)$ is satisfied then the corresponding potential $u(z, {\zb})$ is double periodic with periods~$2$ and~$2\tau$.
\end{rem}

\subsection[Example: O(3) sigma model]{Example: $\boldsymbol{O(3)}$ sigma model}

Let $\G$ be the spectral curve corresponding to a turning point of $N$-particle eCM system.
As we mentioned above, for $N=2{\ell}$ and generic values of the $I_{\rm even}$ (and $I_{\rm odd}=0$) the curve~${\G}$ is smooth of genus $g=N$ and has only two fixed points of the involution \eqref{eq:k3}.
Our construction of solutions for the $O(3)$ sigma model requires the curves with involution having at least $4$ fixed points.

Our next goal is to identify the singular $2{\ell}$-particle eCM curves whose normalization has four fixed points of the involution~\eqref{eq:k3}. Consider the function $F(p,{\alpha}):={\sigma}(z)f(p,{\alpha})$ with $f$ defined by the equation~\eqref{ph}. For even~$N$ it is odd with respect to the involution~\eqref{eq:k3}. The equations
\begin{gather}\label{Fconstr}
\pa_\alpha F(0,0)=0,\qquad \pa_p F (0,0)=0
\end{gather}
cut out an $\left( {\ell}-2 \right)$-dimensional linear subspace in the space of parameters $I_{\rm even}$. If the equations~\eqref{Fconstr} hold, then the expansion of $F$ near the point $(p, {\alpha} ) =(0,0)$ has the form:
\begin{gather*}%\label{Fexpans}
F=b_1 p^3+b_2p^2 \alpha+b_3p \alpha^2+b_4 \alpha^3+\cdots,
\end{gather*}
where $b_i=b_i(I)$ are some linear functions of the parameters $I_{\rm even}$. For generic values of these parameters the equation
\[ b_1 g^3+b_2g^2+b_3 g +b_4=0\]
has three distinct roots $g_i$, $i=1,2,3$. That implies that the equation $F(p,\alpha)=0$ has three distinct solutions of the form
\begin{gather*}%\label{pseries}
p=g_i\alpha +\sum_{s=1}^{\infty}g_{i,s}\alpha^{(2s+1)}, \qquad i=1,2,3.
\end{gather*}
In other words: under the constraints \eqref{Fconstr} the equation $R(k,\alpha)$ has three smooth branches given by
\[
k=-\zeta (\alpha)+p(\alpha),
\]
which represent the expansion of $k$ near three points on the normalization of the singular spectral curve. Since $g_i$ are distinct the branches are invariant under the involution. The normalization of the singular spectral curve has genus $(2\ell-3)$. The corresponding quotient-curve is of genus~$(\ell-2)$.

By Theorems~\ref{thm:main} and~\ref{prop} a pair of curves with involution having four fixed points on each of them, a pair of admissible divisors and an orthogonal $3\times 3$ matrix~$G$ gives a solution to~$O(3)$ sigma-model: the equations~\eqref{constG} are trivially
satisfied, since the matrices $W_{\a}^{\pm}$ are unit matrices as all the branch points are among the preimages of one point $\alpha=0$.
Our next goal is to present this solution to~$O(3)$ model explicitly in theta functions.

On $\G$ the corresponding BA function has $g+1$ poles at the points $\g_1,\dots,\g_{g+1}$ of the admissible divisor~$D$, i.e., the divisor such that~$D$ and~$\sigma (D)$ are preimages of the zeros of a differential ${\rm d}\Omega$ on the factor-curve $\Gamma/\sigma$ with simple poles at the marked point $P$ and three other points~$p_1$,~$p_2$,~$p_3$.

The space of the BA function, i.e., the functions having exponential singularity at $P$ with $g+1$ poles is two-dimensional. Hence on $\G$ the Baker--Akhiezer function in Lemma \ref{lm:main} is a linear combination
\begin{gather}\label{BA0}
{\psi}(z,{\zb}; p) = {\phi}_{1}(z,p) + c(z,\zb) {\phi}_{2}(z,p)
\end{gather}
of two basic functions. The first one equals
\begin{gather}
\phi_{1} = {\rm e}^{z {\Omega}_{2}(p)} \frac{{\theta}(A(p)-A(p_3)-Z) \theta(A(\gamma_1)+Z) \theta(A(\gamma_2)+Z)}
{{\theta}(A(p_3)+Z) \theta(A(p)-A(\gamma_1)-Z) {\theta}(A(p)-A(\gamma_{2})-Z)}\nonumber \\
\hphantom{\phi_{1} =}{} \times \frac{\theta(A(p)+A(p_3)-A(\gamma_1)-A(\gamma_2)+ z U -Z)}
{\theta(A(p_3)-A(\gamma_1)-A(\gamma_2)+ zU -Z)} ,\label{BAapr1}
\end{gather}
with
\begin{gather*}%\label{Zapr}
Z=\sum_{s=3}^{g+1} A({\gamma}_{s})-{\CalK}.
\end{gather*}
${\phi}_{1}$ is the unique function with exponential singularity at the marked point $P$ with poles at $D$ and zero at $p_3$, normalized so that its regular factor at~$P$ equals~$1$.

The function ${\phi}_2$ is defined by the similar conditions: the poles are at~$D$, its regular factor at~$P$ equals~$0$ and ${\phi}_{2} (z,p_3)=1$:
\begin{gather}
\phi_2 = {\rm e}^{z ( {\Omega}_{2}(p)-{\Omega}_{2}(p_3) )}
\frac{\theta(A(p)-Z) \theta(A(p_3)-A(\gamma_1)-Z) \theta(A(p_3)-A(\gamma_2)-Z)}
{\theta(A(p_3)-Z) \theta(A(p)-A({\gamma}_{1})-Z) \theta({A}(p)-A({\gamma}_{2})-Z)} \nonumber\\
\hphantom{\phi_2 =}{} \times \frac{{\theta}(A(p)-A({\gamma}_{1})-A({\gamma}_{2})+ zU-Z)}
{{\theta}(A(p_3)-A(\gamma_1)-A(\gamma_2)+ zU -Z)} .\label{BAapr2}
\end{gather}
Denote by
\begin{gather}
f_1(z)=\phi_1(z,p_1),\qquad f_2(z)=\phi_1(z,p_2),\nonumber\\
g_1(z)=\phi_2(z,p_1),\qquad g_2(z)=\phi_2(z,p_2)\label{fg-functions}
\end{gather}
the values of the basic functions at the points $p_1$ and $p_2$
(their values at $p_3$ are $0$ and $1$, respectively).
The vanishing of the sum of the residues of matrix $\Vert {\phi}_{a}\phi_{b}^{\tilde\sigma} {\rm d}\Omega \Vert$, for $a,b=1,2$ implies
\begin{gather}\label{fg1}
r^2_1 f_1^2 + r^2_2 f_2^2 = 1 , \\
\label{fg2}
r^2_1g_1^2+r^2_2g_2^2 = - r^2_{3} ,\\
\label{fg3}
r^2_1f_1g_1+r^2_2f_2g_2 = 0 ,
\end{gather}
where{\samepage
\begin{gather*}%\label{rrr}
r_i^2=-\operatorname{res}_{p_i} {\rm d}\Omega , \qquad i=1,2,3 .
\end{gather*}
Also from the residue theorem we have $1=r_1^2+r_2^2+r_3^2$.}

The equations~(\ref{fg1})--(\ref{fg3}) imply
\begin{gather}\label{g-to-f}
g_1= \frac{h}{r^2_1} f_2 ,\qquad g_2=-\frac{h}{r^2_2} f_1, \qquad h=\sqrt{-1} r_{1}r_{2}r_3.
\end{gather}
From \eqref{g-to-f} it follows that the values ${\psi}_j = {\psi} ( z,\zb; p_{j} )$, for $j=1,2,3$ are equal to
\begin{gather}
{\psi}_{1} = f_{1} + c \frac{h}{r^2_1} f_2 , \qquad {\psi}_{2} = f_{2} - c \frac{h}{r^2_2} f_1 , \qquad {\psi}_{3} = c .
\label{psi123}
\end{gather}
It is easy to verify, using (\ref{psi123}), (\ref{eq:ccb}), that the vector
${\bf y} = ( y_{j} = r_{j} {\psi}_{j} )_{j=1}^{3}$, obeys, first
\begin{gather*}
 y_{1}^{2} + y_{2}^{2} + y_{3}^{2} = 1 ,
 \end{gather*}
and second
\begin{gather}\label{x-w}
v := \frac{y_{1} + {\ii} y_{2}}{1+y_{3}}={r_1} f_1 (z) + \ii r_2 f_2 (z).
\end{gather}
The functions $f_1(z)$ and $f_2(z)$ are meromorphic functions of $z$, so is $v=v(z)$.

\begin{rem}It is easy to verify that the ansatz
\begin{gather*}
q^{1} + {\ii} q^{2} = \frac{2 v}{1+v{\bar v}} , \qquad q^{1} - {\ii} q^{2} = \frac{2 {\bar v}}{1+v{\bar v}} , \qquad
q^{3} = \frac{1- v{\bar v}}{1+v{\bar v}}
\end{gather*}
with ${\pa}_{\zb} v = {\pa}_{z} {\bar v} = 0$, obeys \eqref{conformal}, i.e., provides a $w_{\infty}$-harmonic map.
\end{rem}
The function $c(z,{\zb})$ in the equation~\eqref{BA0} is found from the equation~\eqref{gluing}. Explicitly, this is the condition that the values ${\psi}_j = {\psi} ( z,{\zb}; \omega_{j} )$, for $j=1,2,3$ are equal to the corresponding values of the BA function on the curve $\bar \G$. In the real case the conditions have the form
\begin{gather*}%\label{gluing_real}
{\bf y}=G{\bar{\bf y}}
\end{gather*}
with matrix $G$ such that
\begin{gather}\label{Greal}
G^{-1}=G^t=\bar G=1.
\end{gather}
If $G=1$, then the gluing condition for $\psi_3$ is just the equation $c=\bar c$, while the reality equations for $y_1$ and $y_2$ give
\begin{gather}
c = - \frac{r^2_1}{h} \frac{f_{1} - {\bar f}_{1}}{f_{2} + {\bar f}_{2}} = \frac{r^2_2}{h}
\frac{f_{2} - {\bar f}_{2}}{ f_{1} + {\bar f}_{1}}.\label{eq:ccb}
\end{gather}
\begin{rem}
We stress that the formulae above hold for any solutions of the $O(3)$-model given by the Theorem~\ref{thm:main}. In general they are quasi-periodic.
\end{rem}
If the periodicity constraints are satisfied, then the functions $\psi_j$ are periodic or anti-periodic. For the curves constructed at the beginning of this section, i.e., the curves which are the normalizations of the singular spectral curves defined by the equation~\eqref{Fconstr} the functions~${\psi}_j$ are elliptic functions. Hence the function $v(z)$ is an elliptic function. In other words we have produced an instanton family of the solutions of the~$O(3)$ model. Let us compute its instanton charge~-- the degree of the corresponding map~$\Sigma\to S^2$.

From the equations~(\ref{BAapr1}), (\ref{BAapr2}) we see that the poles of the functions $f_1(z)$, $f_2(z)$ are the zeros of the function
\[ {\theta}(z U+Z_{0})=0, \qquad Z_0 = A(p_3)-A(\g_1)-A(\g_2)-Z.\]
These are exactly the coordinates $z_i$ of the turning point. Hence, $f_1$ and $f_2$ have $2\ell$ common poles.

Notice that the equation~\eqref{fg1} implies
\begin{gather}\label{ww}
v^{-1}(z)={r_1} f_1 (z) -\ii r_2 f_2 (z).
\end{gather}
Hence, the function $v(z)$ has poles at exactly half of poles of the functions $f_i$. The other half are zeros of~$v(z)$. Hence we get that the instanton number of solutions defined by spectral curve corresponding to the turning points of $2\ell$-particle eCM curves satisfying the constraints~\eqref{Fconstr} is equal to $\ell$ and $v(z)$ has the form
\begin{gather}
v(z) = A \prod_{i=1}^{\ell} \frac{{\sigma}(z - a_{i})}{{\sigma}(z - b_{i})},\label{eq:w-z}
\end{gather}
where the parameters $a_i, b_i \in E$ obey
\[
\sum_{i=1}^{\ell} ( a_{i} - b_{i} ) = 0.
\]
These parameters are determined by the curve $\G$ and the vector $Z$ in~(\ref{BAapr1}),~(\ref{BAapr2}).
\begin{rem}
The NV flows change the parameters $A$, $a_i$, $b_i$ in a complicated manner. In terms of~$\G$,~$Z$ these flows preserve~$\G$ and change $Z$ linearly.
\end{rem}

Let us count the number of parameters in our construction. First, we have $\ell-2$ complex parameters for the curves satisfying~\eqref{Fconstr}. The
space of differentials ${\rm d}\Omega$ on a factor curve with the simple poles at the branch points defining the admissible divisors is of dimension~$\ell$.
Note, that the corresponding potential $u(z,\zb)$ depends only on the equivalence class of the admissible divisor. Indeed, if $\tilde D$ is equivalent to~$D$, then the relation between the BA functions is just
${\tilde \psi}(z,{\zb} ; p) = {\chi}(p) {\psi}(z,\zb; p)$, where ${\chi}(p)$ is the unique function with poles at~$\tilde D$, zeros at~$D$ and normalized by the equation ${\chi}(P)=1$. Notice that the corresponding vectors $\tilde {\bf y}(z,\zb)$ and ${\bf y}(z,\zb)$ defined in~\eqref{y} are equal, thanks to the equation \[ {\rm d}\tilde\Omega=ff^{\tilde\sigma}{\rm d}\Omega. \]
Thus, overall we have $2\ell-3$ complex parameters for the curve and the equivalence class of the admissible divisor. In addition, we have $3$ real parameters for the matrix $G$ satisfying~\eqref{Greal}, i.e., the space of constructed potentials is of real dimension $4\ell-3$.

Recall the well-known formula
\begin{gather*}%\label{instanton}
u(z,\zb)=\pa_z\pa_{\zb}\ln ( 1+ v{\bar v} )
\end{gather*}
for the $\ell$-instanton potentials, where $v(z)$ is of the form~\eqref{eq:w-z}. Although $v(z)$ depends on
$2\ell$ complex parameters, the space of the corresponding potentials $u(z, {\zb})$ is of real dimension $4\ell-3$, since $u$ is invariant
under the ${\rm SU}(2)/\{ \pm 1 \} \approx {\rm SO}(3)$ fractional linear transformations{\samepage
\begin{gather}\label{frac_linear}
{\tilde v}(z)=\frac{c v(z)+d}{-\bar d v(z)+\bar c} , \qquad |c|^2+|d|^2=1.
\end{gather}
Hence, we reproduce all instanton potentials.}

For $G$ fixed, the solution of the $O(3)$-model is defined by a matrix $H$ obeying~\eqref{GH}. The ambiguity in the choice of $H$ is the transformation, equivalent to~\eqref{frac_linear}:
\begin{gather*}%\label{ambiguityH}
H'=h H,\qquad h\in O(3,{\BR}) .
\end{gather*}

\subsection{Other components}
Consider now the $(2\ell-1)$-dimensional locus of turning points for which the spectral curve passes through the point $k=0$, $\alpha=\omega_i$ for some $i=1,2,3$. The space of the corresponding spectral curves is of dimension $\ell-1$ and in terms of the parameters $I_{\rm even}$ is defined by the equation
\begin{gather}\label{cmsing}
f({\zeta}({\omega}_{i}),\omega_{i})=0,
\end{gather}
where the function $f$ is defined by the equation~\eqref{ph}.

Note that by the parity arguments the number of branches of a spectral curve passing through the point $\left( k=0,{\alpha}={\omega}_{j}\right)$ is even, i.e for generic values of~$I_{\rm even}$ satisfying~\eqref{cmsing} the corresponding curve is singular and has one node. Let $\Gamma$ be a normalization of that curve. It is a smooth curve of genus $2\ell-1$ with involution $\sigma$ having four fixed points. Two, including the marked point $P$ over $\alpha=0$ and two over $\alpha=\omega_j$ that are preimages of the node. In that case the monodromy matrices~$W_\a^\pm$ in the equation~\eqref{Wconst} are not-trivial, cf.~\eqref{mon13}. This implies that the matrix $G$ of the Lemma~\ref{lm:main} is diagonal.

The corresponding solutions of $O(3)$-model are given by the same formulae as above: $y_j=r_j\psi_j$ with $\psi_j$ defined in the equations~\eqref{psi123} and (\ref{BA0})--(\ref{BAapr2}). The parameters of the solutions are: $\ell-1$ complex parameters for the spectral curve, $\ell-1$ complex parameters for the equivalence class of an admissible divisor. Overall we have $4\ell-4$ real parameters.

The twisting parameters of the solutions depend on the choice of~$\omega_j$. For example, for $j=1$ the functions $f_1(z)$, $f_2(z)$ defined in the equation~\eqref{fg-functions} have the following monodromy properties
\begin{gather}\label{mon13}
\begin{split}
&f_1(z+\omega_1)=f_1(z), \qquad f_1(z+\omega_2)=-f_1(z), \\
&f_2(z+\omega_1)=f_2(z), \qquad f_2(z+\omega_2)=-f_2(z).
\end{split}
\end{gather}
Hence, the function $v(z)$ defined in \eqref{x-w} has the following monodromy
properties:
\begin{gather*}%\label{monw13}
v(z+\omega_1) = v(z) , \qquad v(z+\omega_2)=-v(z).
\end{gather*}
The functions $f_1(z)$, $f_2(z)$ has $N=2\ell$ poles in the fundamental domain with periods~$\omega_1$,~$\omega_2$. From the equation~\eqref{ww} it follows that $w(z)$ has poles only at half of them. Therefore, it is of the form
\begin{gather*}%\label{wb13}
v(z)=C \prod\limits_{i=1}^{\ell} \frac{\sigma (z-c_i)}{\sigma(z-d_i)}
\end{gather*}
for some $c_i$, $d_i$ satisfying the equation
\begin{gather*}%\label{ab}
2\sum_{i=1}^\ell (c_i-d_i)=\omega_2.
\end{gather*}

\subsection[O(2n+1) sigma model]{$\boldsymbol{O(2n+1)}$ sigma model}

We conclude the section by describing the curves giving the solutions to $O(2n+1)$-model for an arbitrary $n$. Recall that the spectral $2\ell$-particle eCM curve with involution is defined by the equation, cf.~\eqref{ph}:
\[
f ( k+\zeta (\alpha), \alpha ) = 0 , \]
with $I_{\rm odd}=0$. For generic values of the parameters $I_{\rm even}$ it is a smooth curve of genus $2\ell$ with two fixed points of involution over $\a=0$. Suppose now that these parameters satisfy the equations
\begin{gather}\label{last}
\pa_\alpha^i\pa_p^j ( {\sigma} ({\alpha}) f(p,\alpha) ) = 0 , \qquad i+j\leq 2n .
\end{gather}
Of course, in order for \eqref{last} to have solutions we should take $\ell$ sufficiently large.

Then the point $(k=0,{\alpha}=0)$ is a singular point of the spectral curve through which pass $2n+1$ branches. For generic $I_{\rm even}$ the system~\eqref{last} of $(n+1)(2n+1)$ linear equations this is the only singular point of the spectral curve, thus the genus of the normalization of the spectral curve equals $2\ell-2n^2-n-1$.

\begin{rem}
As in the case of the quantum mechanical models \cite{Nekrasov:2018pqq} whose classical limit is an algebraic integrable system, our solutions have a mysterious connection to ${\CalN}=2$ supersymmetric gauge theories in four dimensions. The spectral curves are the Seiberg--Witten curves, the differential $p_{x} {\rm d}p_{y}$ is essentially the Seiberg--Witten differential. In particular, the periodic solutions of the~$O(N)$ model corresponding to reducible Fermi curves are connected to the ${\CalN}=2^{*}$ theory~\cite{Donagi:1995cf}. The origin and the implications of this connection remain to be seen.
\end{rem}

\pdfbookmark[1]{References}{ref}
\LastPageEnding

\end{document}